\documentclass[12pt]{article}%

\usepackage{amsfonts}
\usepackage{amssymb}
\usepackage{setspace}
\usepackage[round]{natbib}
\usepackage{amsmath}%
\usepackage{amsthm}%
\usepackage{relsize}

\usepackage{palatino}
\usepackage{paralist}

\usepackage[top=8pc,bottom=8pc,left=8pc,right=8pc]{geometry}

\usepackage[plain]{algorithm}
\usepackage{algorithmic}

\newtheorem{theorem}{Theorem}

\newtheorem{corollary}[theorem]{Corollary}

\newtheorem{lemma}[theorem]{Lemma}

\usepackage{psfrag,epsf}
\usepackage[dvipdfm]{graphicx}
\usepackage{enumerate}
\usepackage[round]{natbib}
\usepackage{url} 
\usepackage{epstopdf}
\usepackage{subfigure}

\usepackage{booktabs}

\title{Tensor decomposition with generalized lasso penalties}
\author{Oscar-Hernan Madrid-Padilla \\
	James G. Scott } 
\date{This version: \today}

\begin{document}
	
	\maketitle
	
	\begin{abstract}
		We present an approach for penalized tensor decomposition (PTD) that estimates smoothly varying latent factors in multi-way data.  This generalizes existing work on sparse tensor decomposition and penalized matrix decompositions, in a manner parallel to the generalized lasso for regression and smoothing problems.  Our approach presents many nontrivial challenges at the intersection of modeling and computation, which are studied in detail.  An efficient coordinate-wise optimization algorithm for (PTD) is presented, and its convergence properties are characterized.  The method is applied both to simulated data and real data on flu hospitalizations in Texas.  These results show that our penalized tensor decomposition can offer major improvements on existing methods for analyzing multi-way data that exhibit smooth spatial or temporal features.
		
		\bigskip
		
	\end{abstract}
	
	\newpage

\section{Introduction}
\label{sec:intro}

\subsection{Structure and sparsity in multiway arrays}

In recent years there has been an increasing interest in the use of penalized methods for matrix and tensor decompositions.  As in classical principal-components analysis (PCA), the goal of these methods is to represent a high-dimensional data matrix or multiway array in terms of a lower-dimensional set of latent factors.  This line of work differs from classical techniques, however, in the use of penalty functions that encourage these estimated factors to be sparse, structured, or both.  As many previous authors have demonstrated, such regularized estimators usually exhibit a favorable bias-variance tradeoff, particularly when the size of the array far exceeds the number of samples.  They can also make the estimated factors themselves much more interpretable to practitioners.

Existing methods for penalized matrix decompositions have been shown to outperform classical PCA in discovering patterns in application areas such as genomics and neuroscience.  Penalties that encourage structure (such as the fused lasso) provide interpretable results when there is a natural order of the measurements, while penalties that encourage sparsity are useful when there is no such ordering \citep{witten2009penalized}.  In the high-dimensional tensor setting however, existing decomposition methods only enforce sparse constraints.  We address this gap by proposing a method for penalized tensor decomposition (PTD) that allows arbitrary combinations of sparse or structured penalties along different margins of a data array.



Given a data array $Y = \{Y_{lts}\}$, the statistical problem that we study is to find a low-dimensional factor representation (also known as a Parafac decomposition) such that the factors are constrained to be sparse and/or smooth.  For ease of presentation, we restrict attention to the three-way case, but the generalization of our approach to arrays with more than three modes is straightforward.

More explicitly, suppose we are given a set of observations $y_{l,t,s}$, the elements of a three dimensional tensor
$\underline{Y}$ $\in$ $\mathbb{R}^{L\times T\times S}$, that have been generated from the complete tensor model
\begin{equation}
\label{the_model}
y_{l,t,s}=  \sum_{j=1}^J d_{j}^*\, u_{lj}^*\circ v_{tj}^*\circ w_{sj}^* + e_{l,t,s},\,\,\,\,,\,\, l \in \{1,\ldots,L\},\,\,t \in \{1,\ldots,T\},\,\,\,s \in \{1,\ldots,S\}
\end{equation}
with unknown hidden vectors $u_{:j}^*$ $\in$ $\mathbb{R}^{L},$
$v_{:j}^*$ $\in$ $\mathbb{R}^{T},$ $w_{:j}^*$
$\in$ $\mathbb{R}^{S}$, $j = 1, \ldots, J$ and scalars $d_{j}^*,$
$j = 1, \ldots, J$. We will later discuss the missing data problem. For simplicity we assume that the variance $\sigma^{2}$ of the error term $e_{l,t,s}$ is
known and equal to $1$. Moreover, when $J=1$  we suppress the index $j$.   Our goal is to estimate these latent factors, which can be challenging since we only have one  observation for each combination $u_{lj}^*,$ $v_{tj}^*,$
$w_{sj}^*$. However, we assume that this task is aided by the presence of special
structure in these true vectors.  Explicitly, we assume that some of the vectors $\{u^*_{\cdot,j} \}_{j=1}^{J},\{v^*_{\cdot,j}\}_{j=1}^{J},$, $\{w^*_{\cdot,j} \}_{j=1}^{J} $ are restrictions of smooth functions defined in the interval $[0,1]$. For instance, it might be the case that $u^*_{lj} = u^*_{j}(l/L)$ for $l = 1,\ldots,L$, where $u^*_{j}$ is a piecewise continuous or differentiable function on $[0,1]$.


A natural situation in which this would arise is when one of the modes of the data array corresponds to a temporal or spatial axis.  Our main contribution is to provide optimization algorithms for finding Parafac decompositions that shrink towards such structure.  To do so, we apply a generalized lasso penalty along each mode of the array.  We refer to this class of methods as penalized tensor decompositions (PTD). 

We face two main challenges in estimating the factors.  First, the resulting optimization problem is non-convex. We propose to reach a stationary point using block coordinate descent, as in \citet{allen2012sparse}, and we provide convergence rates for a single-block udpate.  This leads us to the second challenge: unlike in the sparse unconstrained problem formulated by \cite{allen2012sparse}, for our case of a generalized lasso penalty, it is not clear how to make the block-coordinate updates. Our results provide a novel way of doing so that exploits the multi-convex structure of the problem, and that provides efficient algorithms for finding the factors when formulating the problem either in a penalized or constrained form.

\subsection{Relation to previous work}

Structurally constrained estimation is an active area of research, and we do not attempt a comprehensive review. Our work draws heavily in the one dimensional case on advances in understanding the one dimensional case, where penalized regression has been widely studied in the literature \citep{friedman2010applications,kim2009ell_1,tibshirani1996regression,tibshirani2005sparsity}.  For instance, in protein mass spectroscopy and gene expression data measured from a microarray, the fused lasso has been used to obtain interpretable results \citep{tibshirani2005sparsity}. The fused lasso is a natural choice here, since it encourages neighboring measurements to share the same underlying parameter.   Similarly, to enforce smoothness in the solution, trend filtering has been proposed \cite{kim2009ell_1} as a way to place one-dimensional function estimation within the convex optimization framework.  The trend filtering penalized-regression problem has found applications in areas as diverse as image processing and demography.

In the case of matrix decomposition, the need for penalized methods arises in applications in genetic data, where there are multiple comparative genomic hybridizations and we expect correlation among observations at genetic loci that are close to each other along the chromosome.  As shown in \cite{witten2009penalized}, by considering different choices of penalties, we can recover different kinds of structures along either the rows or the columns of a data matrix.  See the references in \cite{witten2009penalized} for a much more comprehensive bibliography on sparse principal components analysis.

In moving from matrices to multiway arrays, Parafac decompositions offer an attractive  framework for recovering latent lower dimensional structure. This is due to their easy interpretability as well as feasibility of computation   \citep{anandkumar2014guaranteed,harshman1970foundations,karatzoglou2010multiverse,kolda2009tensor,kroonenberg2008applied}. More generally, Tucker models have been proposed as general models for multiway data and have been successfully applied in many areas \citep{cichocki2013tensor}. Other popular methods for tensor decompositions include those described in \cite{bhaskara2013smoothed} and \cite{de2000multilinear}. However, these approaches do not provide structural or sparse solutions. This point was made by \cite{allen2012sparse}, who proposed a sparse penalized Parafac decomposition method that outperforms the classical Parafac decomposition when the true solutions are sparse. More recently, \cite{sun2015provable} also considers sparse tensor recovery and provides statistical guarantees for such a task.

In this paper, we study methods for structured, as opposed to sparse, tensor factorizations. Our approach is inspired by the penalized matrix decomposition methods from \cite{witten2009penalized}. We  generalize the matrix-decomposition problem to the framework of tensor Parafac decompositions while incorporating solution algorithms for a more broad class of penalties, including trend filtering for factors that are smooth (e.g in space or time). 

\subsection{Basic definitions}
\label{preliminaries}
We now introduce notation and definitions used throughout the paper. This material
can be found in \cite{cichocki2009nonnegative}, to which we refer the reader for more details.  Let $I_{1},$$I_{2}$..., $I_{N},$ denote index $N$
upper bounds. A tensor $\underline{Y}$ $\in$ $\mathbb{R}^{I_{1}\times I_{2}\times...\times I_{N}}$
of order $N$ is an $N-$way array where elements $y_{i_{1},i_{2}...,i_{N}}$
are indexed by $i_{n}$ $\in$ $\left\{ 1,2,....,I_{n}\right\} ,$
for $n$ $\text{= }$ $1,...,$$N.$
Tensors are denoted by capital letters with a bar, e.g.~$\underline{Y}$ $\in$ $\mathbb{R}^{I_{1}\times I_{2}\times...\times I_{N}}$.
Matrices are denoted by capital letters, e.g~$Y$, and for a matrix $Y$ we denote by $Y^{-}$ its generalized inverse.  Vectors are denoted by lower case letters, e.g~$y$. The outer product of two vectors $a$ $\in$ $\mathbb{R}^{I}$
and $b$ $\in$ $\mathbb{R}^{J}$ yields a rank-one matrix $A=a\circ b=ab^{T}\in\mathbb{R}^{I\times J}$,
and the outer product of three vectors $a$ $\in$ $\mathbb{R}^{I},$
$b$ $\in$ $\mathbb{R}^{J}$ and $c$ $\in$ $\mathbb{R}^{Q}$ yields
a third-order rank-one tensor$A=a\circ b\circ c\in\mathbb{R}^{I\times J\times Q}$. We use  $\|\cdot \|_{F} $ to indicate the usual Frobenius norm of tensors.  The mode-$n$
multiplication of a tensor $\underline{Y}$ $\in$ $\mathbb{R}^{I_{1}\times I_{2}\times...\times I_{N}}$
by a vector $a$ $\in$ $\mathbb{R}^{I_{n}}$ is denoted by $Z:=\underline{Y}\times_{n}a\in\mathbb{R}^{I_{1}\times\ldots\times I_{n-1}\times I_{n+1}\times\ldots\times I_{N}}$,
and element-wise we have $z_{i_{1}...i_{n-1}i_{n+1}\ldots i_{N}}=  \sum_{i=1}^{I_n} y_{i_{1}i_{2}...i_{N}}a_{i_{n}}  $.
\subsection{Outline}

The rest of the paper is organized as follows. Section 2 defines our statistical approach to rank-1 tensor decompositions based on generalized lasso penalties. Section 3.1 then provides solution algorithms for our problem formulation when the penalties are used to define a set of constraints on the parameters. This is done by exploiting the efficiency of solution-path algorithms for generalized-lasso regression problems. In Section 3.2, we then study an unconstrained version of the problem where the penalties enter directly into the objective. Because the original problem is not convex, this is not equivalent to the constrained formulation, and some important algorithmic differences are highlighted.  After developing algorithms for rank-1 tensor decompositions, Section 3 concludes by extending these ideas to the general case of multiple factors.

Section 4 presents a convergence analysis for our fundamental rank-1 decomposition algorithm.  In Section 5, using simulated data, we benchmark against state-of-the-art methods on rank-1 and multiple factor decompositions, measuring the error of recovery with the Frobenius norm.  We then validate our algorithms on two real data sets involving flu hospitalizations in Texas and motion-capture data. Finally, Section 6 present a brief discussion of the overall framework proposed in this paper.
\section{Penalized tensor decompositions}
\label{the_problem}

We first consider the case $J=1$. Taking a point of view similar to  \cite{witten2009penalized}, for positive constants $c_u$, $c_v$  and $c_w$,  we formulate the following problem:
\begin{equation}
\label{rank_1_problem}
\begin{aligned}
& \underset{u\in\mathbb{R}^{L}, v \in \mathbb{R}^{T}, w\in\mathbb{R}^{S}, g \in \mathbb{R} }{\text{minimize}}
& & \|Y-g\, u\circ v\circ w
\|_{F}^{2} \\
& \text{subject to}
& &  \| D^{u}u\|_1 \,\leq\,c_u,\quad \|D^{v} v\|_1\,\leq\,c_v,\quad \|D^{w}w\|_1 \,\leq\,c_w\\
& & &  u^{T}u \,=\,1, \quad v^{T}v \,=\,1, \quad w^{T}w \,=\,1 \,,
\end{aligned}
\end{equation}
where $D^{u}$, $D^{v}$ and $D^{w}$ are matrices which are designed to enforce structural constraints. When the context is clear we will suppress the superscript and simply use the notation $D$. We note that an alternative, although non-equivalent, formulation is based on an unconstrained version of  (\ref{rank_1_problem}) given as
\begin{equation}
\label{rank_1_problem2}
\begin{aligned}
& \underset{u^{T}u \,=\,1,\, v^{T}v \,=\,1,\, w^{T}w \,= \,1}{\text{minimize}}
& &  \|Y-g\, u\circ v\circ w
\|_{F}^{2}  + \lambda_u\, \|D^{u} u\|_1 + \lambda_v\, \|D^{v} v\|_1  + \lambda_w\, \|D^{w} w\|_1 \, ,
\end{aligned}
\end{equation}
with the same unit-norm constraints on the factors.   In Section 3, we will discuss the computational differences between these formulations in detail.

We now briefly discuss a broad class of penalties of potential interest to practitioners. We focus on choices that penalize first- and higher-order differences  in each factor, which correspond to the fused lasso and trend filtering, respectively \citep{tibshirani2011solution}. The fused lasso penalty was suggested in \cite{witten2009penalized} to detect regions of gain for sets of genes in matrix-decomposition problems.  For this penalty, the associated $D$ matrix is the $(S-1) \times S$ first-difference matrix, $D_{i,j} = 1$ if $j=i$,  $D_{i,j} = -1$ if $j= i+1$ and $D_{i,j} =0$ otherwise. 
As discussed in \cite{tibshirani2011solution}, this penalty gives a piecewise-constant solution to linear-regression problems, and it is used in settings where the coordinates in the true model are closely related to their neighbors.  A related choice for $D$ is oriented incidence matrices of a graph; see, e.g.~\cite{arnold2014efficient}.  These are constructed as generalizations of the 1-dimensional fused lasso on an underlying graph $G$.


Still other choices for $D$ correspond to polynomial trend filtering, which impose a piecewise polynomial structure on the underlying object of interest. These are constructed as follows. First define the polynomial
trend filtering of order $1$ as $D_{tf,1} \in \mathbb{R}^{(S-2) \times S} $ where $D_{tf,1}  = (D^{(1)})^T\,D^{(1)}$ and $D^{(1)} \in \mathbb{R}^{(S-1) \times S}$ is the first order difference matrix.
Then, recursively construct the polynomial trend filtering matrix of order $k$ as $D_{tf,k}=D_{1,d}\cdot D_{tf,k-1}$.

The polynomial trend filtering fits (especially for $k$ $=$ $3$)
are similar to those that one could obtain using regression splines
and smoothing splines, However, the knots (changes in kth derivative)
in trend filtering are selected adaptively based on the data,
jointly with the inter-knot polynomial estimation \citep{tibshirani2011solution}. A comprehensive
study of polynomial trend filtering can be found in \cite{tibshirani2014adaptive}.
We note that Problem (\ref{rank_1_problem2})  was already studied in \cite{allen2012sparse} for the case in which all the matrices $D^{u}$, $D^{v}$  and $D^{w}$  are set to be the identity. This is the case of having the  L1 penalty on each mode. The authors in \cite{allen2012sparse}  proposed a fast algorithm to solve the problem.  However, the L1 penalty has the disadvantage of encouraging only sparsity. If the true factors are not sparse but instead locally flat or smooth, then having sparse constraints on the factors performs poorly. This phenomenon was observed in \cite{witten2009penalized} in the context of matrix decompositions, where the fused lasso penalty was shown to properly recover flat vectors in the factors of the decomposition when the L1 penalty failed to do so. We will extend these ideas to tensor decompositions, applying penalties from the generalized lasso class.  We now turn to the question of how to fit these models efficiently.


\section{Solution algorithms}

\subsection{Constrained problem}

Since (\ref{rank_1_problem}) is a non-convex problem, we propose to consider a block coordinate-descent routine. However, in order to have convex block-coordinates-updates, we instead state the following problem:
\begin{equation}
\label{rank_1_problem3}
\begin{aligned}
& \underset{u\in\mathbb{R}^{L},v \in\mathbb{R}^{T},w\in\mathbb{R}^{S}}{\text{maximize}}
& & \underline{Y} \times_{1} u \times_{2} v \times_{3} w \\
& \text{subject to}
& &  \| D^{u}u\|_1 \,\leq\,c_u,\quad \|D^{v} v\|_1\,\leq\,c_v,\quad \|D^{w}w\|_1 \,\leq\,c_w\\
& & &  u^{T}u \,\leq\,1, \quad v^{T}v \,\leq\,1, \quad w^{T}w \,\leq\,1.
\end{aligned}
\end{equation}
This differs from (\ref{rank_1_problem}) in two ways. First, the objective has been reformulated in a more convenient way, but it is easy to show that this results  in an equivalent problem \citep{kolda2009tensor}. Secondly, the unit norm constraints have been relaxed to the convex constraints that each factor fall into the unit ball.  Additionally, following \cite{witten2009penalized},  a simple modification can naturally handle missing data. Denoting by $M$ the set missing observations, we solve the missing data problem by replacing the objective function in (\ref{rank_1_problem3}) with the function
\begin{equation}
F(u,v,w) = \sum_{(l,t,s)  \in \{1,\ldots,L\}\times  \{1,\ldots,T\} \times  \{1,\ldots,S\} - M }  Y_{l,t,s}\,u_{l}\,v_{t}\,w_{s}  \\
\end{equation}

Note that (\ref{rank_1_problem3}) has a multilinear objective function in $u$,  $v$, and $w$.  Since the penalties induced by $D^{u}$, $D^{v}$  and $D^{w}$ are convex, we can use coordinate-wise optimization in order to solve this problem. For example, when $v$ and $w$ are fixed, the update for $u$ is found by solving the following problem:
\begin{equation}
\label{updating_u}
\underset{u}{\text{maximize}}\;\, \left(  \underline{Y} \times_{2} v \times_{3} w\right)^{T}u\qquad \text{subject to}\quad\ \|u\|_{2}^{2}\,\leq1\,, \; \|D^{u}u\|_1 \leq c_{u}.
\end{equation}
It would seem that a solution to (\ref{updating_u}) would not in general have unit norm.  But it is possible to ensure that this will be the case---that is, to ensure the solution follows on the boundary of the $\ell^2$ constraint set---as long as $c_u$ is chosen properly based on the KKT conditions.  A similar phenomenon was observed for the matrix case in \cite{witten2009penalized}.  One of our results is that the solution to (\ref{updating_u}) will very often turn out to have unit norm, despite our convex relaxation. A rigorous statement of this result will be given later.

Our strategy to solve (\ref{rank_1_problem3})  is to sweep  through the vectors iteratively by proceeding with block coordinates updates. Thus starting from initials  $u^0$, $v^0$  and $w^0$, we proceed by solving, at iteration $m$, the problems shown in Algorithm \ref{alg1}.  It should be pointed out here that the best we can hope with Algorithm 1 is to obtain a local minimum to (\ref{rank_1_problem3}). It will be shown latter with  our experiments that this local minimum provides interpretable and accurate estimators. Note that while the algorithm is structurally quite simple, the individual block-coordinate updates are non-trivial to solve efficiently. The remainder of this section discusses how this can be done.

\begin{algorithm}[t]                    
	\caption{Constrained problem block coordinate descent}          
	\label{alg1}                           
	\begin{algorithmic}                    
		\STATE   \[
		\begin{array}{lll}
		u^m  & = & \underset{u}{\text{arg min} } \left\{  \left(  -\underline{Y} \times_{2} v^{m-1} \times_{3} w^{m-1}\right)^{T}u\qquad \text{subject to}\quad\ \|u\|_{2}^{2}\,\leq1\,, \; \|D^{u}u\|_1 \leq c_{u}. \right\} \\
		v^m  & = & \underset{v}{\text{arg min} } \left\{  \left(  -\underline{Y} \times_{1} u^{m} \times_{3} w^{m-1}\right)^{T}v\qquad \text{subject to}\quad\ \|v\|_{2}^{2}\,\leq1\,, \; \|D^{v}v\|_1 \leq c_{v}. \right\} \\
		w^m  & = & \underset{w}{\text{arg min} } \left\{  \left(  -\underline{Y} \times_{1} u^{m} \times_{2} v^{m}\right)^{T}w\qquad \text{subject to}\quad\ \|w\|_{2}^{2}\,\leq1\,, \; \|D^{w}w\|_1 \leq c_{w}. \right\} \\
		\end{array}
		\]
		
	\end{algorithmic}
\end{algorithm}

Given the symmetry of the problem, without loss of generality, we focus on
the update for $u$. We notice that the constraint set involves a non-differentiable function,
implying that it is not possible to use a gradient-based method.  Before describing our approach, we first discuss two natural possibilities and explain why they were ultimately rejected.

First, a simple approach is to include a slack variable $z = D^{u}u$ and use the ADMM algorithm.  However, the resulting update for $u$ would require solving a constrained problem using, for example, an interior-point method. This rapidly becomes infeasible, since it requires solving a large dense linear system.

A second natural approach is to use the novel ADMM algorithm from \cite{zhu2015augmented}  to solve each of the block-coordinate updates. For instance, the update for $u$ would involve solving the problem
\begin{equation}
\label{alternative_ADMM}
\begin{array}{llll}
u^m  & = & \underset{u}{\text{arg min} }  & \left(  -\underline{Y} \times_{2} v^{m-1} \times_{3} w^{m-1}\right)^{T}u\qquad  \\
& &  \text{subject to} &   \quad\ \|u\|_{2}^{2}\,\leq1\,, \; \|z\|_1 \leq c_{u},\,\,\, \,\,z= D^{u}u,\,\,\,\, (E_u - (D^{u})^TD^{u})^{1/2}  u = \tilde{z} \, , 
\end{array}
\end{equation}
where $E_u$ is a matrix such that $E_u \succeq  (D^{u})^TD^{u} $. Then proceeding as in  \cite{zhu2015augmented}, we observe that (\ref{alternative_ADMM})  can be solved in linear time, as the update for $u$ is a simple projection on the unit $\ell_2$ ball, while the update for $z$ requires projecting in a $\ell_1$ ball with the algorithm  from \cite{duchi2008efficient}.  (The actual updates for our problem are given in the appendix.)  However, while this algorithm indeed solves the constrained-problem updates, we find in that practice the ADMM routine requires a long time to converge. In particular, it presents problems enforcing the constraint that $\|D^u u^m\|_1 \leq c_u$, so that the solution returned after reasonable runtimes is actually quite far from the feasible region.

This motivates us to consider a different approach to solve the block-coordinate updates in (\ref{alg1}). We appeal to the following theorem, which suggests a simple method and also implies that, typically, the solution lies on the boundary of the unit ball. That is, it satisfies the non-convex constraint of problem (\ref{rank_1_problem}), despite our relaxation.

\begin{theorem}
	\label{main_theorem}
	Assume that $c_u$ $>$ $0$ and  $\underline{Y} \times_{2} v \times_{3} w \notin Range\left((D^{u})^{T}\right)$. Then the solution
	to (\ref{updating_u}) is given by
	\begin{equation}
	\label{solution}
	u^{*}=\frac{\left(-\underline{Y} \times_{2} v \times_{3} w   -(D^{u})^{T}\hat{\gamma}_{\lambda^{*}}\right)}{\|-\underline{Y} \times_{2} v \times_{3} w - (D^{u})^{T}\hat{\gamma}_{\lambda^{*}}\|_{2}}
	\end{equation}
	where
	\begin{eqnarray}
	\label{eqn:dual_argument}
	\hat{\gamma}_{\lambda} &=& \underset{\| \gamma\| _{\infty}\leq\lambda}{\arg\min}\quad\frac{1}{2}\|  -\underline{Y} \times_{2} v \times_{3} w-(D^{u})^{T}\gamma\| _{2}^{2} \\
	\lambda^{*} &=& \underset{0\leq\lambda}{\arg\min}\left[\|-\underline{Y} \times_{2} v \times_{3} w -(D^{u})^{T}\hat{\gamma}_{\lambda}\| _{2}+\lambda c_{u}\right].
	\end{eqnarray}
\end{theorem}

As a direct consequence of the proof of Theorem \ref{main_theorem}, we can solve (\ref{updating_u}) by first solving (\ref{eqn:dual_argument}) with the solution-path algorithm from \cite{tibshirani2011solution}, then finding $\lambda^*$ and finally $u^{*}$. The explicit algorithm is given in the appendix.

Unfortunately, there  is no characterization available of of the computational
time to compute the solution path. It is only known the cost at each iteration is $O(L)$ in
its worst case, but it is unknown how many kinks $K$ that a particular problem will have.
Moreover, we notice that after the solution path is computed, the next two steps require
$O(K L)$ cost. Therefore, the total cost for updating $u$ is $O(K L)$. 

\subsection{Unconstrained version}

The framework we have introduced for rank-1 approximations has some nice features.  In particular, the choice of tuning parameters is more intuitive, since this directly imposes a constraint on the smoothness of the solutions. However, the optimization routine derived from Theorem \ref{main_theorem} is computationally intensive. In particular, for large dimensions of the penalty matrices, computing the entire solution path can still be somewhat slow. To avoid this, we revisit (\ref{rank_1_problem2}) and consider a problem equivalent to its convex relaxation:
\begin{equation}
\label{rank_1_problem4}
\begin{aligned}
& \underset{u\in\mathbb{R}^{L},v \in\mathbb{R}^{T},w\in\mathbb{R}^{S}}{\text{minimize}}
& & -\underline{Y} \times_{1} u \times_{2} v \times_{3} w + \lambda_u\, \|D^{u} u\|_1 + \lambda_v\, \|D^{v} v\|_1  + \lambda_w\, \|D^{w} w\|_1 \\
& \text{subject to}
& & u^{T}u \,\leq\,1, \quad v^{T}v \,\leq\,1, \quad w^{T}w \,\leq\,1 \, .
\end{aligned}
\end{equation}
As in the constrained case, we solve (\ref{rank_1_problem4})  via block-coordinate updates. Now the update for $u$  is obtained by solving
\begin{equation}
\label{updating_u2}
\underset{u}{\text{minimizing}}\;\,\, -\left(  \underline{Y} \times_{2} v \times_{3} w\right)^{T}u  +  \lambda_u\|D^{u}u\|_1 \qquad \text{subject to}\quad\ \|u\|_{2}^{2}\,\leq1\,.
\end{equation}

The solution to (\ref{updating_u2}) can be characterized in the same manner as for the constrained case. In fact, the proof of Theorem \ref{main_theorem} implies the following corollary:

\begin{corollary}
	\label{main_corolary}
	With the notation and assumptions from Theorem (\ref{main_theorem}), the solution to
	\begin{equation}
	\label{problem_for_corolary}
	\underset{u\in\mathbb{R}^{S}}{\text{minimize}}  -\left(  \underline{Y} \times_{2} v \times_{3} w\right)^{T}u  + \lambda\,\|D^{u}u\|_1\,\,\,\,\,   \text{subject to} \,\,\,  \|u\|_{2}^{2} \leq 1 
	\end{equation}
	has the following form, where $\hat{\gamma}_{\lambda}$ is defined in (\ref{eqn:dual_argument}):
	\begin{equation}
	\label{solution}
	u^{*}=\frac{\left( -\left(  \underline{Y} \times_{2} v \times_{3} w\right)^{T}-(D^{u})^{T}\hat{\gamma}_{\lambda}\right)}{\|   -\left(  \underline{Y} \times_{2} v \times_{3} w\right)^{T}-(D^{u})^{T}\hat{\gamma}_{\lambda}\|_{2}} \, .
	\end{equation}
	
\end{corollary}

An interesting consequence of  the closed-form formula (\ref{solution}), and the proof of Theorem \ref{main_theorem}, is that we can solve (\ref{updating_u2}) by  first solving a generalized lasso problem and then projecting the solution into the unit ball. Explicitly, we first find
\begin{equation}
\label{ADMM_update}
\hat{u} = \underset{u\in\mathbb{R}^{L}}{\text{arg min}}\left\{ \|u -  \underline{Y} \times_{2} v \times_{3} w \|_{2}^{2}  + \lambda\,\|D^{u}u\|_1,\,\,\,,\,\text{subject to} \,\,\,  \|u\|_{2}^{2} \leq 1\right\}
\end{equation}
and  $\hat{u}/\|\hat{u}\|_2$ becomes the solution to (\ref{updating_u2}). Therefore, for trend-filtering problems, we can solve the regression problem step with the fast ADMM algorithm from  \cite{ramdas2014fast}. Moreover, for the case of a fused lasso penalty, the update for $u$  can be done in linear time \citep{johnson2013dynamic}. Because these two algorithms are so efficient, the penalized formulation from (\ref{rank_1_problem4}) can be solved much more cheaply than the constrained formulation from (\ref{rank_1_problem3}).

\subsection{A toy example}

We illustrate the advantage of problem (\ref{rank_1_problem4}) over the formulation
from (\ref{rank_1_problem3}) using a toy example. We consider $ u^* \in  \mathbb{R}^{10}$ and $w^* \in \mathbb{R}^{400}$ as the size of $v^*$ varies. Here, $u^*$ and
$w^*$ are as in Structure 2 in Figure \ref{fig:structure_examples},  while $v^*$ is the function $\cos(9 \,\pi\,t)$ evaluated at evenly
spaced locations in $[0, 1]$.  Taking initial values from the power method, we compare the solutions
from one iteration of the unconstrained formulation when choosing the penalty parameters adaptively,
versus an ``oracle'' version of the constrained problem with $(c_u,c_v, c_w) = (\|D^{u}u^*\|_1, \|D^{v}v^*\|_1, \|D^{w}w^*\|_1)$.  This choice of hyperparameters for the constrained problem is obviously optimal, but requires knowledge of the true factors, and is therefore unrealistic in practice.

Figure \ref{time} demonstrates the favorable trade-off offered by the unconstrained formulation with adaptively chosen tuning parameters. We observe that while the constrained formulation algorithm based on the solution-path computation is the most accurate, the unconstrained formulation is competitive in terms of reconstruction error, and much more efficient. The ADMM algorithm based on \cite{zhu2015augmented} is substantially less accurate than the other two methods.


Moreover, in practice it would be necessary to solve the constrained
problem with more than one value of the tuning parameters, since we do not know $\|D^{u}u^*\|_1$.   Hence the penalized version is strongly preferred: we can do adaptive parameter choice more cheaply than solving the constrained version for a single hyperparameter setting, without a major loss of performance even under an optimal hyperparameter choice.

\begin{figure}[t]
	\subfigure[]{\includegraphics[width=7.5cm]{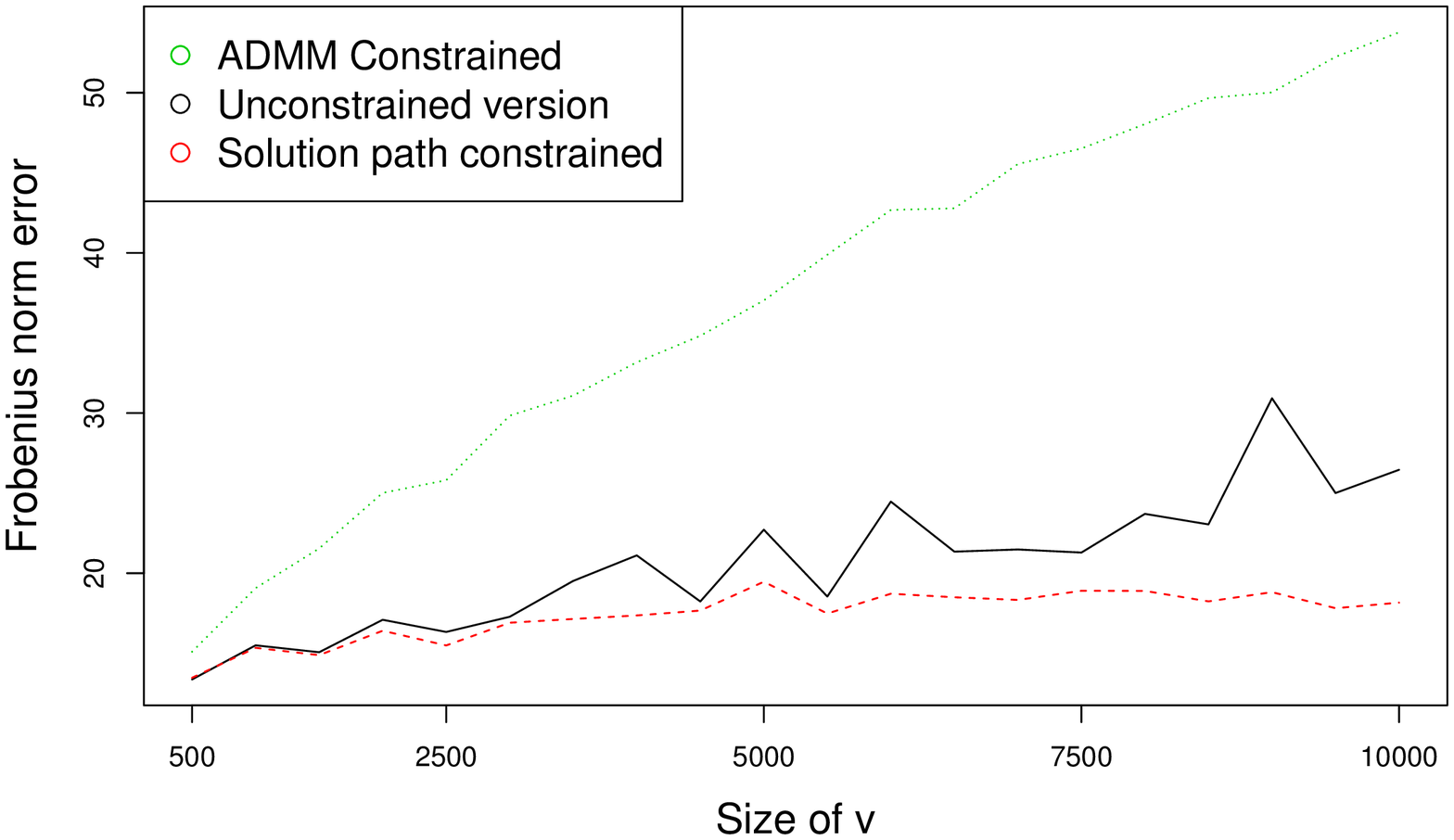}\hspace{5px}}
	\subfigure[]{\includegraphics[width=7.5cm]{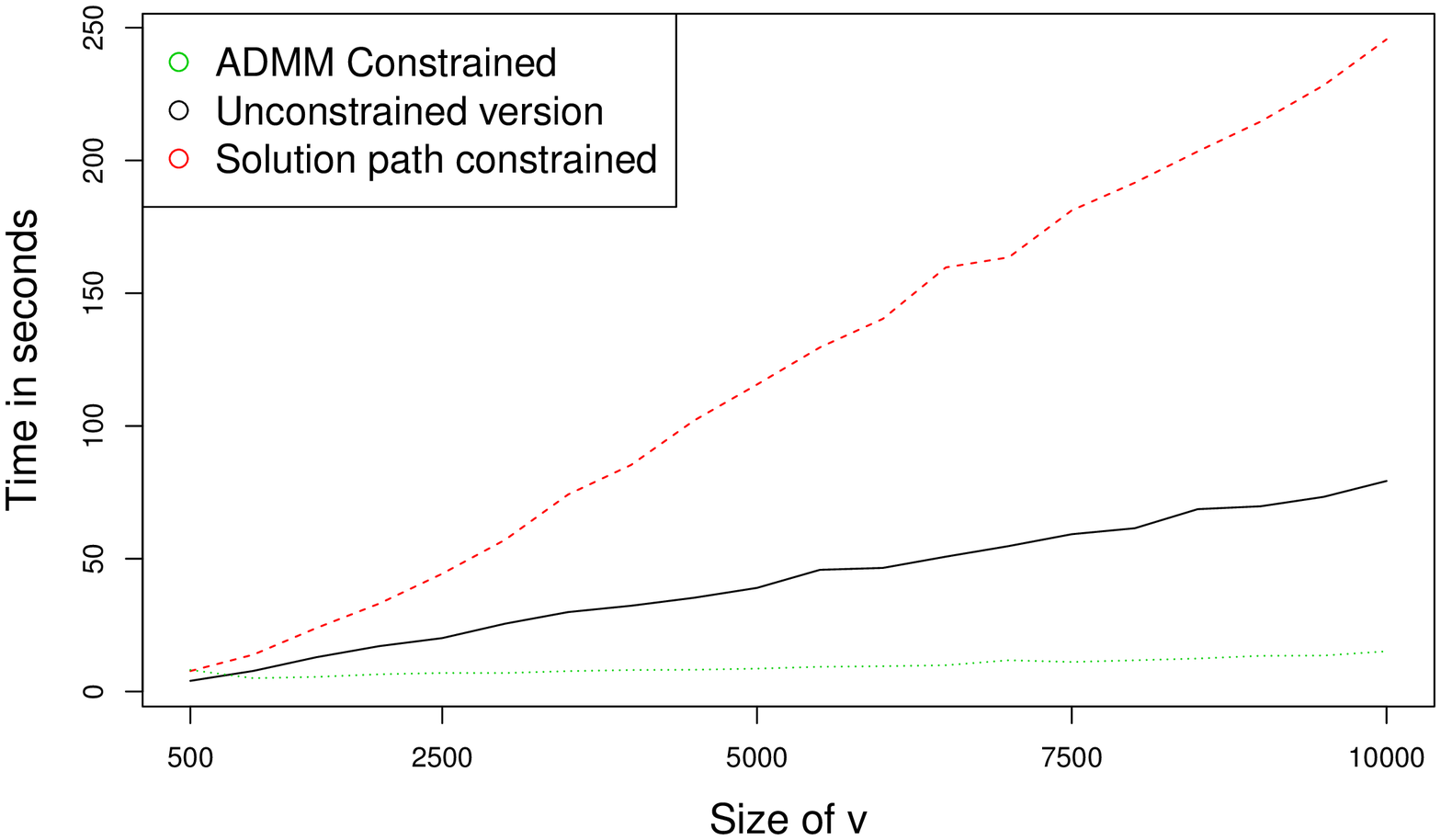}  \hspace{10px}}\\
	\caption{\label{time} Panel (a): Frobenius error comparison of the  of three different methods for finding a rank-1 decomposition. These are: Algorithm \ref{alg1} with the ADMM method from \cite{zhu2015augmented}, block coordinate descent for solving the unconstrained problem (\ref{rank_1_problem4}), and  Algorithm \ref{alg1} using the solution path method as described in Section 3.1.  Panel (b): For each of the methods, time in seconds for solving one problem with a particular choice of tuning parameters.   Our unconstrained formulation with adaptive chosen penalties achieves nearly the reconstruction error of the unconstrained formulation with optimal hyperparameter choice, but at far less computational cost.}
\end{figure}

With regards to the choice of regularization, we can consider two alternatives based on
cross validation. The first of these follows \cite{witten2009penalized}. This procedure involves randomly deleting a percentage of the input data and solves the problem on the resulting tensor. The estimated
tensor produces predicted values on the deleted entries, allowing one to compute mean square
error of prediction for these notionally missing values. The parameters $\lambda_u$, $\lambda_v$  and $\lambda_w$ are then chosen to minimize the prediction
error. This is particularly attractive when multiple processors are available, given that
independent problems with different tuning parameters can be solved in parallel.

The other alternative for cross validation applies to (\ref{rank_1_problem4}) and it is based on adaptively choosing the tuning
parameters. Thus, before estimating each vector (say $u$), we obtain a generalized lasso
regression problem and hence we can choose $\lambda_u$ by cross validation. We randomly separate
the coordinates of the response vector into training and test set, solving the problem in the
training set and computing the mean squared error of the predicted solution on the test set.
This exploits the fact that $u$ is a smooth function, and therefore given a solution based on the training set, we can provide
estimates at the locations in the test set by interpolation.

\subsection{Multiple factors }
\label{multiple_factors}

In the case of multiple factors, the main difference of the tensor case versus the matrix case is that we must find all the factors jointly \citep{kolda2009tensor}, as opposed to estimating factor $k+1$ using the residual from the fitted $k$-factor model. Fortunately, it is straightforward to use any of the algorithms in the previous section to handle multiple factors. Hence, to estimate the factors in (\ref{the_model}), we state the problem
\begin{equation}
\label{multiple_factor_2}
\begin{aligned}
& \underset{u_j, v_j, w_j  }{\text{minimize}}
& & \|\underline{Y}- \sum_{j =1}^J d_j\,u_j\circ v_j \circ w_j \|_{\text{F}}^{2} + \sum_{j=1}^{J} \left[ \lambda_{u,j} \|D_j^u\,u_j\|_{1}  +  \lambda_{v,j} \|D_j^v\,v_j\|_{1} +  \lambda_{w,j} \|D_j^w\,w_j\|_{1}  \right]   \\
& \text{subject to} &&  \|u_j\|_{2}^{2} \leq 1 \quad \|v_j\|_{2}^{2} \leq 1 \quad \|w_j\|_{2}^{2} \leq 1 \quad j = 1,...,J,  \\
\end{aligned}
\end{equation}
where the matrices $D_j^{u}$,$D_j^{v}$  and $D_j^{w}$  are  chosen to capture different structural features desired for the solutions. Here, $\lambda_{u,j}$, $\lambda_{v,j}$  and $\lambda_{w,j}$ are tuning parameters. Now we solve (\ref{multiple_factor_2}) by starting with initial guesses $\{u^{j}\}$,  $\{v^{j}\}$,  $\{w^{j}\}$, $\{d^{j}\} $ and applying the iterative updates listed in Algorithm \ref{alg2} exploiting results from Section 3.2.

\begin{algorithm}[t]                      
	\caption{Multiple factors}          
	\label{alg2}                           
	\begin{algorithmic}                    
		\STATE   Loop for $j_0 =1:J$,
		\[
		\begin{array}{lll}
		u^{j_0} & \leftarrow & \underset{ \|u\|_{2}^{2} \leq 1}{\text{arg min}} \,\,   \left\|u -\,  \underline{Y}\times_2 v^{j_0}\times_3 w^{j_0} + \sum_{j \neq j_0} d^j\,(v^{j_0})^T\,v^j\,(w^{j_0})^T\,w_j\,u^j \right\|_{2}^{2}  +  \lambda_{u,j_0}\,\|D_{j_0}^u\,u \|_{1}, \\
		v^{j_0} &  \leftarrow & \underset{ \|v\|_{2}^{2} \leq 1}{\text{arg min}} \,\,   \left\|v -\,  \underline{Y}\times_1 u^{j_0}\times_3 w^{j_0} + \sum_{j \neq j_0} d^j\,(u^{j_0})^T\,u^j\,(w^{j_0})^T\,w_j\,v^j \right\|_{2}^{2}  +  \lambda_{v,j_0}\,\|D_{j_0}^v\,v \|_{1},  \\
		w^{j_0} &  \leftarrow & \underset{ \|w\|_{2}^{2} \leq 1}{\text{arg min}} \,\ \left\|w -\,  \underline{Y}\times_2 u^{j_0}\times_3 v^{j_0} + \sum_{j \neq j_0} d^j\,(u^{j_0})^T\,u^j\,(v^{j_0})^T\,v_j\,w^j \right\|_{2}^{2}  +  \lambda_{w,j_0}\,\|D_{j_0}^w\,w \|_{1}. \\
		d^{j_0} &  \leftarrow & \underline{Y}\times_1\,u^{j_0}\times_2\,v^{j_0}\times_3\,w^{j_0}  - \sum_{j \neq j_0} d^j\,(u^{j_0})^T\,u^j\,(v^{j_0})^T\,v^{j}\,(w^{j_0})^T\,w^j.
		\end{array}
		\]
		\STATE End loop
	\end{algorithmic}
\end{algorithm}

In practice the number of latent factors can be chosen with an ad-hoc rule by looking at the proportion of the variance explained (as with a scree plot in ordinary PCA). One can look at the solutions provided by different values of $J$. The choice of $J$ then corresponds to the number factors such that the increase in variance explained obtained by solving the problem with more factors is negligible. We illustrate this in our real data example.

Finally, in situations where the number of factors is large, the number of possible combinations of tuning parameters becomes challenging. One possibility to address this is to choose the parameters adaptively as discussed in Section 3.2. Hence, every time a factor is to be updated we select the parameter from a small grid of values. This ensures that, for instance, when dealing with fused lasso penalties each block coordinated update can be done in linear time. On the other hand, a different alternative is to use the same penalty parameter for all the vectors corresponding to the same level of smoothness. For instance, one can use $\lambda_u,j = \lambda_{u,i}$ if $D^{u_j} = D^{u_i}$.  This reduces the burden of cross-validation.

\section{Convergence  analysis}

We now examine the convergence of the block-coordinate algorithms developed in the previous section. Here, we assume that $J=1$ in Model (\ref{the_model}). In this case we recall that the underlying true tensor can be decomposed as the outer product of vectors $u^{*} \in \mathbb{R}^{L},$ $v^{*} \in \mathbb{R}^{T}$ and $w^{*} \in \mathbb{R}^{S},$ times a constant $d^{*}$. Moreover, we assume that the  matrices $D$ are chosen to be either fused lasso or trend filtering penalties. Thus, $D^{u} =  D^{(k_u+1)} \in \mathbb{R}^{(L - k_u)\times L}$,   $D^{v} =  D^{(k_v + 1)} \in \mathbb{R}^{(T - k_v)\times T}$  and $D^{w} =  D^{(k_w+1)} \in \mathbb{R}^{(S - k_w)\times S}$ with $k_u,k_v$  and $k_w$ $\in $ $\{0,1\}$.

Our proof is inspired by  the work on convergence rates for generalized lasso regression problems from  \cite{wang2014trend}. The theorem states that, when starting with good initials, it is necessary to sweep through the data only once. The proof of the claim is based on the identity 
\[
\text{P}(A \cap B \cap C) =  \text{P}(A)  \text{P}(B \mid A)  \text{P}(C \mid A \cap B)
\]
for any events $A$, $B$  and $C$. A related statement can be made in the case of multiple factors for a single update depending on the other factors.  See the result in the appendix; the main difference there involves an error  measurement that depends on the factors taken as fixed.
\begin{theorem}
	\label{convergence_rates}
	Let $\{u^{1},v^{1},w^{1}\}$ denote a one-step update from Algorithm \ref{alg1}, based on initial values $\{u^{0},v^{0},w^{0}\}$, and assume that  $\|D^{u} u^{*}\|_1 \leq c_w$, $\|D^{v} v^{*}\|_1 \leq c_w$, and $\|D^{w}w^{*}\|_1 \leq c_w$.  Then,  there exists a constant $c > 0$  such that if  $t >0 $ satisfies
	\[
	\max\left\{\frac{c\,t}{d^{*}\sqrt{L}}+ \frac{2\,c_u\,L^{k_u +1/2}}{d^{*}}, \frac{c\,t}{d^{*}\sqrt{T}}+ \frac{2\,c_v\,T^{k_v +1/2}}{d^{*}}, \frac{c\,t}{d^{*}\sqrt{S}}+ \frac{2\,c_w\,S^{k_w +1/2}}{d^{*}}  \right\} \leq \frac{1}{2^{5}},
	\]
	and 

	\[
	\|v^{0} - v^{*}\|_2 <  2^{-1/2},\, \|w^{0} - w^{*}\|_2 <    2^{-1/2},
	\]	
	then
	\[
	\begin{array}{l}
	\text{P}\Big( \|u^{1}-u^{*}\|_2^2 \leq 16\left( \frac{c\,t}{d^{*}\sqrt{L}}+ \frac{2\,c_u\,L^{k_u +1/2}}{d^{*}} \right), \|v^{1}-v^{*}\|_2^2 \leq 16\left( \frac{c\,t}{d^{*}\sqrt{T}}+ \frac{2\,c_v\,T^{k_v +1/2}}{d^{*}} \right),  \\
	\quad\quad  \|w^{1}-w^{*}\|_2^2 \leq 16\left( \frac{c\,t}{d^{*}\sqrt{S}}+ \frac{2\,c_w\,S^{k_w +1/2}}{d^{*}} \right)  \Big) \\
	\geq   \Psi(t,L)\,\Psi(t,T)\,\Psi(t,S),
	\end{array}
	\]
	where
	\[
	\Psi(t,x) =  \left(1 -   \sqrt{\frac{2}{\pi}}\frac{1}{t}e^{-\frac{t^2}{2} } - \frac{2^{1/2}}{x^{3/2}\sqrt{5\,\pi\,\log(x)}}  \right).
	\]
	
\end{theorem}

Theorem \ref{convergence_rates}  states that with good initials our rank-1 decomposition algorithm will be very close to the true factors under weak assumptions concerning the smoothness of the true factors. Thus,  in practice before running our algorithms, we can consider a simple initialization that consists od solving Algorithm (\ref{alg1}) for the case where the matrices $D^{u}$,  $D^{v}$ and $D^{w}$  are all zero. This is known as the power method \citep{kolda2009tensor}. Moreover, statistical guarantees for a very related method to this procedure  were studied in \cite{anandkumar2014guaranteed}.

Finally, it should be note that Theorem \ref{convergence_rates} implicitly suggests that an appropriate choice of
tuning parameter is $(c_u, c_v, c_w) = ( \|D^u u\|_1, \|D^v v\|_1, \|D^w w\|_1)$ which only involves the
true latent vectors. In the case of the unconstrained version, a very similar statement
to Theorem \ref{convergence_rates}  holds by taking $\lambda_u = O(L^{k_u+1/2}\,\sqrt{\log(L)})$, $\lambda_u = O(T^{k_v+1/2}\,\sqrt{\log(T)})$ and $\lambda_w = O(S^{k_w+1/2}\,\sqrt{\log(S)})$.%

Finally, we note  that, as one would expect, the larger $d^*$ is, the better we should expect to perform. This is intuitive, given that when $d^{*}$  increases and the unit vectors $u^{*}$, $v^{*}$  and  $w^{*}$  are fixed, the standard Gaussian noise becomes small compared to the magnitude of the observations.

\begin{figure}[t!]
	\begin{center}
		\includegraphics[width=5in,height= 5in]{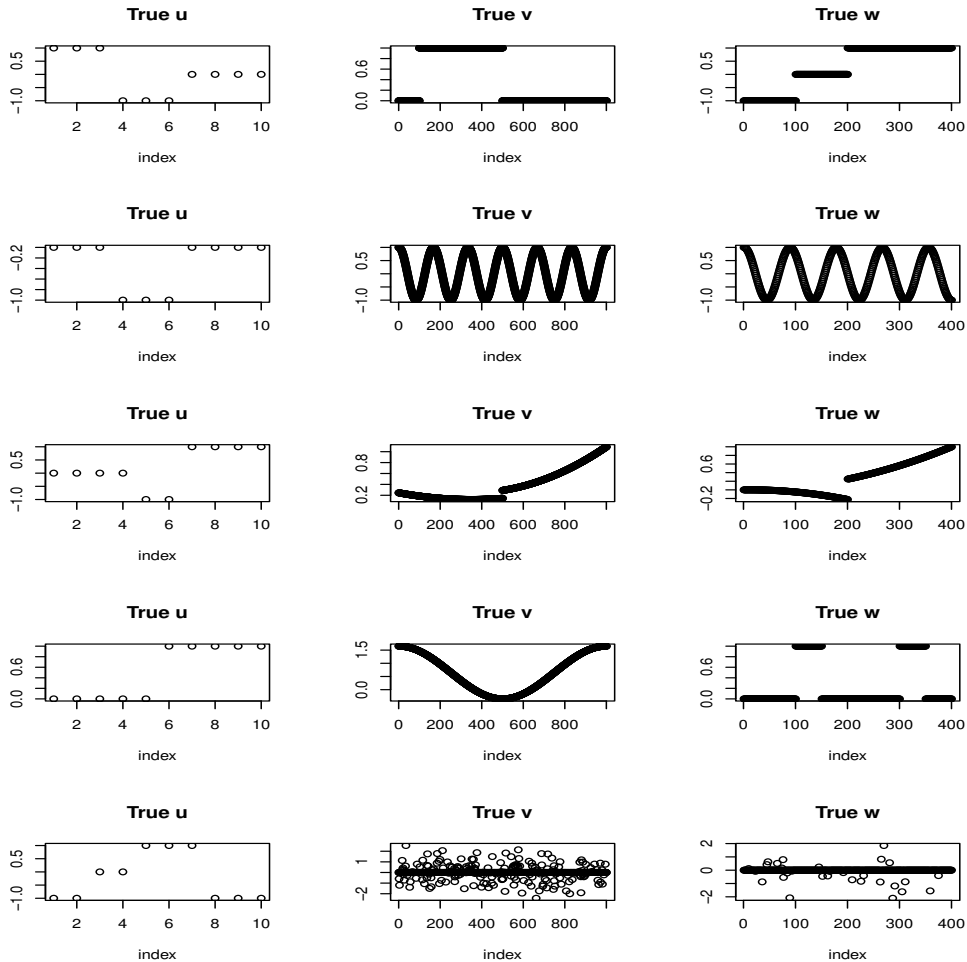} 
		\caption{\label{fig:structure_examples} Latent vectors generating the structures for our examples. Each row gives rise to a different structure by taking the outer product on the corresponding vectors. }
	\end{center}
\end{figure}

\section{Experiments}

Our experiments focus mainly on the task of rank-1 recovery, since all of our algorithms are based on the development of a rank-1 PTD. For all our simulations we use the Frobenius norm of the difference between the estimated and true tensors as a measure of overall accuracy.  The Frobenius norm is a natural choice of model fit, since we also benchmark against a recovery method that does not directly produce a rank-1 tensor but does provide an estimate of the true mean tensor. This method is based on the idea of stacking several penalized matrix decompositions using the technique from  \cite{witten2009penalized}.  Specifically, we consider the tensor of observations $\underline{\tilde{X}}$ as a collection of 10 distinct $1000 \times 400$ matrices, each of which is estimated via a rank-1 PMD. This will lead to 10 estimated rank-1 matrices which are concatenated to build a 10
$\times$ 1000 $\times$ 400 tensor.  We call this procedure, with an abuse of notation,
PMD$\left(P_{v},P_{w}\right)$ where $P_{v}$ and $P_{w}$ are the
penalties on $v$ and $w$, when computing the rank-1 PMD matrices. 

The other methods included in the study are the PTD with different penalties $P_{u},P_{v},P_{w}$ denoted as PTD$\left(P_{u},P_{v},P_{w}\right)$. We consider choices such as the L1 penalty, the fussed lasso (FL) and  
trend filtering of order k (TFk). Note that we are implicitly comparing to the method from \cite{anandkumar2014guaranteed} since, for rank-1 recovery, this reduces to the power method, and hence to PTD(L1,L1,L1) for appropriate parameters.

For our simulations,  the tuning parameters by cross validation on a grid of possible values for each of the  parameters $\lambda_u$, $\lambda_v$, and $\lambda_w$.  In every experiment, we randomly select $10\%$ of the data for testing, using the other $90\%$ as training data. Out of a range of candidate tuning parameters we select those that produce the smallest error on the 10$\%$ held-out set. This process is repeated for each of 100 simulations, for different methods and structures, in order to obtain average Frobenius errors for all the competing methodologies with respect to every structure.

To see how different choices of penalties can behave under different scenarios, we  ran experiments using five different rank-1 tensors as the true mean tensor. These choices are designed to explore a range of plausible structures that we might find in real problems. For the first structure both $v$ and $w$ are piecewise flat.  For the second, both $v$ and $w$ are periodic functions.  For the third, both $v$ and $w$ are piecewise quadratic polynomials.  For the fourth, $v$ is smooth and $w$ is piecewise constant. For the fifth, both $v$ and $w$ are sparse but with no specific structural pattern like smoothness or flatness.  The goal of this final scenario is to understand how structural penalties perform in a data set where they are not warranted.  Further details of this simulation are included in the appendix. Figure \ref{fig:structure_examples} also shows a plot of these different structures.

\begin{table}[t!]
	\begin{footnotesize}
		\begin{center}
			\caption{\label{table_1} Comparison of the Frobenius norm error between the true tensor and the estimated tensor using different methods.}
			\medskip
			\begin{tabular}{l r r r r r}
				Method  & Structure 1  & Structure 2  & Structure 3  & Structure 4  & Structure 5  \\
				\midrule
				PTD(L1,L1,L1)    & 37.37  & 47.63  & 46.16     & 39.91   & 40.58 \\
				PTD(L1,FL,FL)    & 6.31   & 27.54  & 11.76     & 10.30   & 57.15 \\
				PTD(L1,TF1,FL)   & 15.07  & 20.49  & 11.55     & 9.00    & 70.32 \\
				PTD(L1,TF1,TF1)  & 17.61  & 14.40  &11.85      & 12.40   & 79.25 \\
				PMD(L1,L1)       & 85.05  & 89.10  & 100.70    & 91.89   & 72.87 \\
				PMD(L1,FL)       & 49.09  & 50.14  & 52.70     & 22.73   & 92.20\\
				PMD(FL,FL)       & 15.05  & 43.17  &  25.64    & 33.95   & 114.09\\
			\end{tabular}
		\end{center}
	\end{footnotesize}
\end{table}

The results of our simulation study are shown in Table \ref{table_1}. In all cases, PTD converged with few iterations, usually less than 10. From these results, it is clear that different choices of penalty are suitable for different problems.  For structure 1, in which the true $v$ and $w$ are piecewise flat, the combination PTD(L1, FL, FL) outperforms all the other choices that we considered. Interestingly, PTD(L1,TF1,FL) and PTD(L1, TF1,TF1) provided better results than the ``stacking'' method PMD(FL,FL). Note also that PTD(L1,TF1,FL) and PTD(L1,TF1,TF1) behave fairly similar to one another.  This is reasonable since a piecewise constant function is a special case of a piecewise linear function and hence we would expect that TF1 would produce only slightly worse results than fused lasso.

\begin{table}[t!]
	\begin{footnotesize}
		\begin{center}
			\caption{\label{table_2} Comparison of the Frobenius norm error between the true tensor and the estimated tensor using for different levels of noise and a fixed structure, averaging over 100 Monte Carlo simulations}
			\medskip
			\begin{tabular}{l r r r r r}
				Method  & $\sigma = 1.25$  & $\sigma = 1.50$ & $\sigma = 1.75 $ & $\sigma = 2.00 $  & $\sigma = 2.25 $  \\
				\midrule
				PTD(L1,L1,L1)    & 62.66   & 81.66   & 80.46     & 99.50   & 94.37\\
				PTD(L1,FL,FL)    & 32.61   & 38.80   & 41.63     &  46.32  & 49.33\\
				PTD(L1,TF1,FL)   & 24.55   & 28.55   & 32.35     &  37.87  & 38.43 \\
				PTD(L1,TF1,TF1)  & 17.00   & 21.35   & 22.27     &  27.09  & 27.36 \\
				PMD(L1,L1)       & 116.19  & 139.57  & 158.71    & 185.05  & 209.45\\
				PMD(L1,FL)       & 66.80   & 76.81   & 83.65     & 98.18  & 111.09 \\
				PMD(FL,FL)       & 52.43   & 57.52   & 65.36     & 83.98  &  92.71 \\
			\end{tabular}
		\end{center}
	\end{footnotesize}
\end{table}

Moreover, Table \ref{table_1} also illustrates when our methodology should not be expected to work. This is what happens with structure 5, where there is no spatial pattern in the true vectors $u$, $v$ and $w$, and instead they are merely sparse (80$\%$ of their coordinates are zero).  Here, as expected,  PTD(L1,L1,L1) outperforms any of our methods.

In the previous experiment we simulated all data sets with the assumption that the noise had variance 1. Now we fix the rank-1 tensor mean of Structure 2, where both $v$ and $w$ are periodic functions, and then we compare the performance of different methods as the standard deviation of the noise changes. Recalling that in Structure 2 both $v$  and $w$ are periodic smooth, it does not come as a surprise that PTD(L1,TF1,TF1) provides the best performance in all situations considered in Table \ref{table_2}. In addition, it is clear that the error of all methods increases as the variance of the noise does. Nevertheless, the performance of our method seems to tbe the most stable.

\begin{table}[t!]
	\begin{footnotesize}
		\begin{center}
			\caption{\label{table_3} Comparison of the Frobenius norm error between the true tensor and the estimated tensor using different methods, averaging over 100 Monte Carlo simulations}
			\medskip
			\begin{tabular}{l r r r r r r r r r r}
				\toprule
				Method & \multicolumn{6}{r}{ Structures } \\
				&  1,2  &  1,3  & 1,4   &  2,3   & 2,4 & 3,4  &1,2,3  & 1,2,4   & 1,3,4 & 2,3,4       \\
				\midrule
				Anandkumar       & 544.0  & 310.5  & 85.4 & 121.0  & 128.9     & 273.4   & 534.9   & 555.3   & 346.6   &  350.3  \\
				PTD(L1,FL,FL)    & 55.3   & 46.5   & 27.1 & 71.2   & 59.7      & 107.3   & 184.2   & 48.3    & 102.8   &    126.1\\
				PTD(L1,TF1,TF1)  & 51.7   & 71.6   & 67.8 & 49.2   & 50.9      & 94.0    & 120.3   & 75.2   & 141.6   &    120.8\\
			\end{tabular}
		\end{center}
	\end{footnotesize}
\end{table}

Finally, we evaluate the recovery of mean tensors having multiple factors,  with $\sigma = 1$. Scenarios where the true model consists of $J=2$  and $J= 3$ are considered. Our comparisons  are based on taking sums of different rank-1 tensors using the structures discussed before. The competing methods are PTD(L1,FL,FL) and PTD(L1,TF1,TF1), versus Algorithm 1 from \cite{anandkumar2014guaranteed}. For the latter, we set the number of initializations $L = 30$  and the number of iterations $N=10$.  The results in Table \ref{table_3}  show a clear gain for our approach over the method from \cite{anandkumar2014guaranteed}, which do not impose any smoothness constraints on its solutions. 

\section{Real data examples}

\subsection{ Flu hospitalizations in Texas}

As a simple illustrative example, we consider measurements of flu activity and atmospheric conditions in Texas, see the appendix for information how to collect the data. There are 5 variables measured  daily across 25 cities in Texas from January 1, 2003 to December 31, 2009. The variables are: maximum and daily average observed concentration of particulate matter (air quality measure),  maximum and minimum temperature, and a measure of flu intensity capturing flu-related hospitalizations per million people. The data tensor is thus a 5x25x2556 array where we expect clear temporal patterns, along with correlations among the five variables.  For example, during the winter months we would expect an increase in flu-related hospitalizations, correlated with seasonal patterns of maximum and minimum daily temperatures.

To show the kind of interesting results that one can get with our methods, we compute a two-factor Parafac decomposition.  We use trend filtering of order 2 in the temporal mode and no penalty on the other two modes (although it would be straightforward to incorporate a penalty on the spatial mode as well.)  We use our main result (\ref{main_theorem}) to find the factors using coordinate-wise optimization. The tuning parameter for the trend-filtering penalty is chosen by cross validation from a grid of values to ensure that we get a smooth vector for the time mode.


\begin{figure}[t]
	\subfigure[]{\includegraphics[ width=7cm]{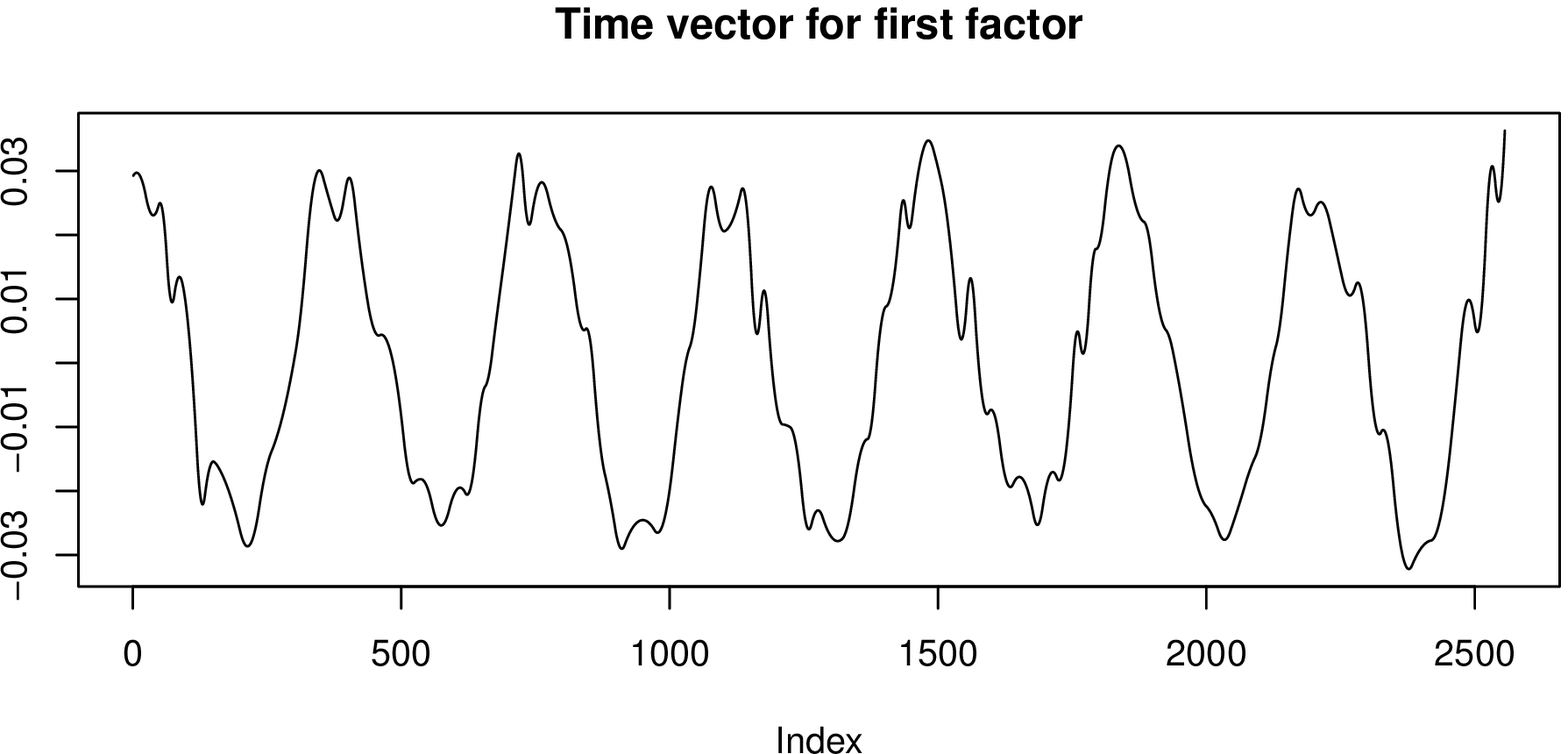} \hspace{5px}}
	\subfigure[]{\includegraphics[width=7cm]{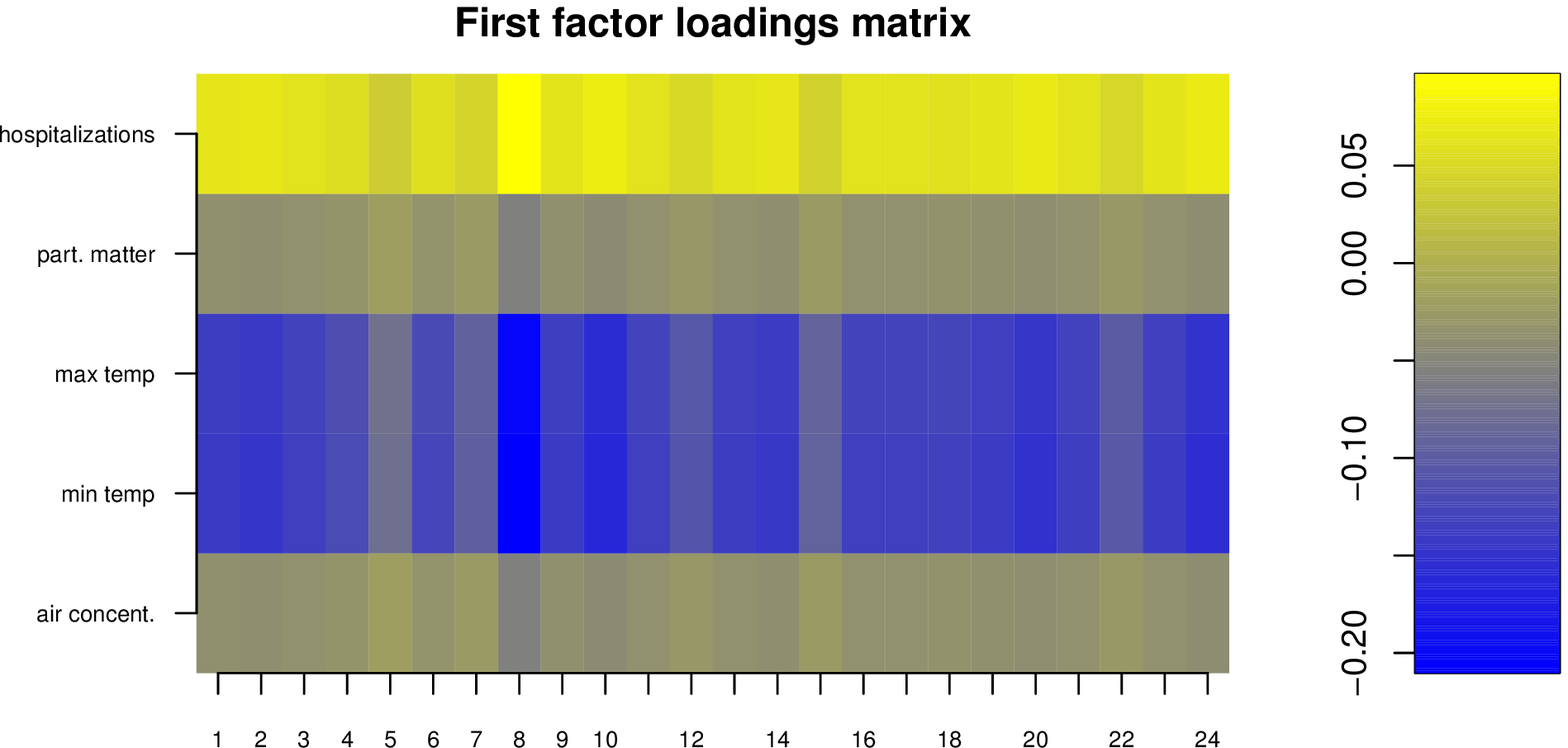} \hspace{10px}}\\
	\subfigure[]{\includegraphics[ width=7cm]{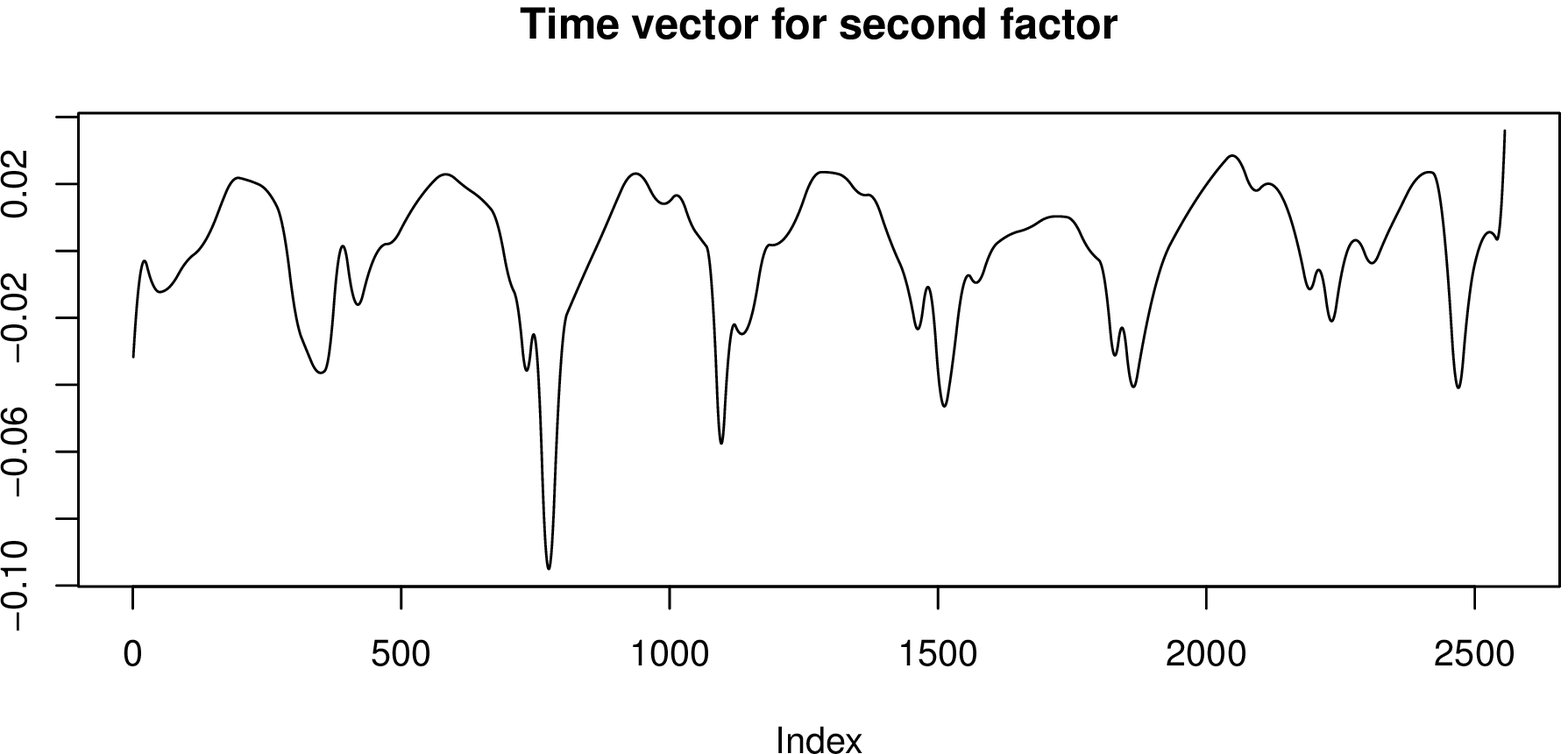} \hspace{10px}}
	\subfigure[]{\includegraphics[width=7cm]{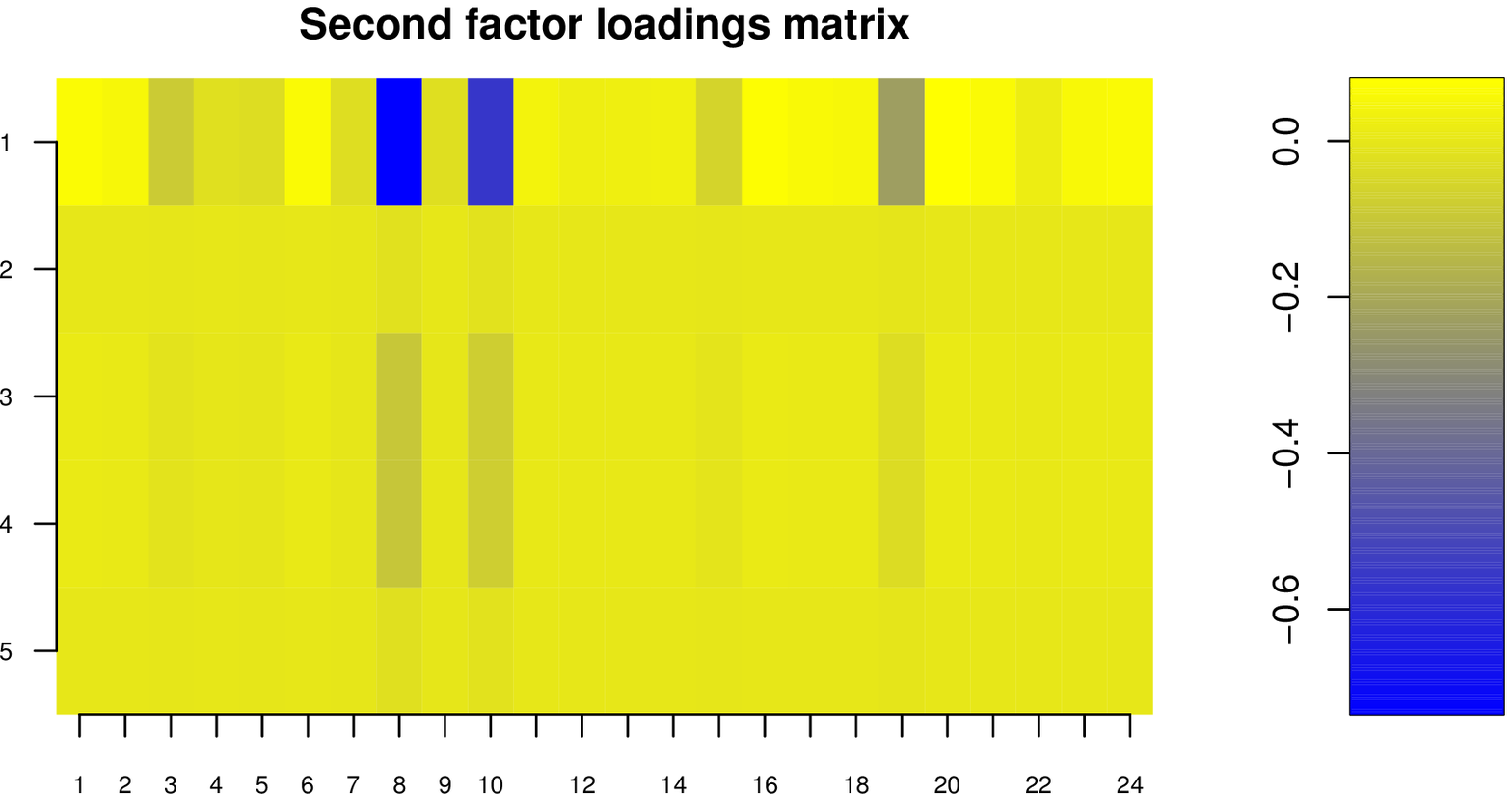} \hspace{5px}}
	\caption{\label{figure_1} (a) Time vector for the first factor (b) Loadings matrix for first factor (c) Time vector for second factor (d) Loadings matrix for second factor.}
\end{figure}

We considered fitting models with different values of $J$, we found that a model with one  factor explains  $36\%$ of the variance, a model with two factors explains $\%45$ percent of the variance, and a model with three factors results in increase in variance explained of less than $\%1$ with respect to the case $J =2$. Moreover, the model with 3 factors results in highly correlated factors. For this reason we use a model with $2$ factors.

From Figure \ref{figure_1} we note a clear seasonal effect. In the first factor we observed that the loadings for the flu intensity, minimum temperature, and maximum temperature can be all explained in a similar way. For the first of these three variable the loadings  are positive in every city. Hence, given the shape of the time vector we see a periodic pattern of flu cases  across cities with the highest during the winter months and the lowest during the summer months. 




\subsection{Motion capture data}

For a more challenging task, we evaluate the performance of our PTD method using data from the motion capture (moCap) repository at \url{mocap.cs.cmu.edu}. This consists of subjects performing different physical activities in repeated independent trials. We construct 3-array tensors by taking sets of videos as one mode, 12 representative variables of the body movements as the second mode, and data frames in time as the third mode. The 12 variables are listed in the appendix.

We built 2 tensors each for 5 different tasks, with each task generating a training-set tensor and a test-set tensor.  The training set tensor corresponds to a single subject performing multiple repetitions of a single related set of physical activities. Similarly, the corresponding test-set tensor corresponds to that same subject performing further repetitions of those same activities. For example, the first data set (comprising 1 tensor in the training set and 1 tensor in the test set) is called 126-swimming; this is formed by looking at 8 videos of subject 126 performing different swimming styles. In the moCap repository, videos 1,3,6,8 are used for training while videos  2,4,7,9 are used for testing. This results in both tensors having dimensions 4$\times$253$\times$12.

The other four data sets, explained in detail in the appendix, are 138-story (subject 138 walking and moving arms); 107-walking (subject 107 walking with obstacles); 9-running (subject 9 running); and 138-marching (just like it sounds).  For these data sets, the tensors dimensions are 4$\times$325$\times$12, 4$\times$828$\times$12, 4$\times$128$\times$12, 4$\times$371$\times$12 respectively.

In this context, our PTD approach can be thought of as a smoothing step applied to the training-set tensor, to yield better out-of-sample predictions for the test-set tensor.  We evaluate the performance of the method by calculating the reconstruction error (again, by Frobenius norm) when using the fitted/smoothed training-set tensor to predict the corresponding test-set tensor.

We find that for the tensors considered here, rank-1 is the best Parafac decomposition, since models with higher factors result in strongly correlated factors.  We ran our rank-1 PTD with a trend-filtering penalty of order 2 on the second mode, and no constraints in the other modes. We compare against the PMD using the same degree of smoothness, as well as the classical PCA method from \cite{anandkumar2014guaranteed}. From Table \ref{mocap} it is clear that PTD offers the best performance.  Thus we can see the  gain of using smooth penalties, reflecting the fact that physical movements involve motion-capture variables that change smoothly in time. Moreover, it is clearly favorable to pool information across videos, as our method does, rather than treating them independently, as with the PMD algorithm.

\begin{table}[t!]
	\begin{footnotesize}
		\begin{center}
			\caption{\label{mocap} Comparison of the Frobenius norm error between the estimated tensor and the test tensor for the moCap datasets}
			\medskip
			\begin{tabular}{l r r r r r }
				\toprule
				Method & \multicolumn{3}{r}{ Task} \\
				&   126-swimming  &   138-story &  107-walking  & 9-running   &  138-marching         \\
				\midrule
				Anandkumar          &  254.80   & 134.63   & 135.17   & 84.40 & 143.86         \\
				PTD(L1,TF2,TF2)     &  250.98   & 131.78   & 134.92   & 84.29 & 142.44     \\
				PMD(L1,TF2,TF2)     &  267.89   & 145.14   & 143.43   & 88.06 & 149.41         \\
			\end{tabular}
		\end{center}
	\end{footnotesize}
\end{table}

\section{Discussion}

In many problems, tensors offer a natural way to represent high-dimensional, multiway data sets. However, tensors by themselves are difficult to interpret, creating the need for methods that shrink towards some simpler, low-dimensional structure.

Parafac models have been widely used for this task, but existing state-of-the-art  methods typically constrain the factors to be orthogonal, or simply do not enforce any constraints. As we have shown, this can be undesirable in practice, especially if one is looking for more interpretable factors, where there is a natural spatial or temporal relation between observation, and it is expected that the factors will be smooth. We fill this gap by providing a set of methods that precisely offer piecewise smooth Parafac decompositions. Our methods exploit state of the art convex optimization algorithms and are shown to have excellent performance in our experiments. We leave for future work the study of algorithms for more general classes of penalties that can potentially be non-convex.

Finally, we have shown two alternatives for finding our smooth tensor decompositions with generalized lasso penalties. The constrained formulation seems to be an attractive option for practitioners, with clear intuitive control over the level of smoothness exhibited by the solutions.  On the other hand, in light of its computational advantages, the unconstrained formulation offers a more practical approach, especially if there is no pre-existing knowledge about the anticipated smoothness of the solutions.

		\begin{small}
			\singlespacing
			\bibliographystyle{abbrvnat}
			\bibliography{tensor}
			
		\end{small}
		
		\appendix

			\section{ADMM algorithm to solve the constrained updates}
			
			In this section we  discuss how to find the updates for Algorithm 1 from the main document using the ADMM algorithm from \cite{zhu2015augmented}. Since these are symmetric we focus on the particular update $u^m$. In this case the problem is
			
			\begin{equation}
			\label{alternative_ADMM}
			\begin{array}{llll}
			u^m  & = & \underset{u}{\text{arg min} }  & \left(  -\underline{Y} \times_{2} v^{m-1} \times_{3} w^{m-1}\right)^{T}u\qquad  \\
			& &  \text{subject to} &   \quad\ \|u\|_{2}^{2}\,\leq1\,, \; \|z\|_1 \leq c_{u}, \\
			& & &  \,\,z= D^{u}u,\,\,\,\, (E_u - (D^{u})^TD^{u})^{1/2}  u = \tilde{z}.
			\end{array}
			\end{equation}
			We define $y$  as
			
			\[
			y = \underline{Y} \times_{2} v^{m-1} \times_{3} w^{m-1}
			\] 
			and solve  (\ref{alternative_ADMM}), using the ADMM algorithm from \cite{zhu2015augmented}, by considering the iterative updates
			
			\[
			\begin{array}{lll}
			u_{k+1} &=& \underset{ \|u\|_2^2 \leq 1  }{\arg \min} \left\{  \frac{1}{2}\|y  - u\|_2^2  +  2\,(\alpha_{k} - \alpha_{k-1})^T\,D^{u}\,u + \frac{\rho}{2}(u - u_k)^T\,E_u\,(u - u_k)    \right\}  \\
			&  = &  \frac{ \frac{y}{2}  - (D^u)^T\,(\alpha_{k} - \alpha_{k-1}) +  \frac{\rho}{2}E_u\,u_k  }{\|\frac{y}{2}  - (D^u)^T\,(\alpha_{k} - \alpha_{k-1}) +  \frac{\rho}{2}E_u\,u_k\|_2 } \\
			z_{k+1} &= & \underset{\|z\|_1 \leq c_u}{\arg \min}  \left\{ \|  D^u\,u_{k+1} + \rho^{-1}\,\alpha_k - z\|_2^2 \right\}  \\
			\alpha_{k+1} &= &  \alpha_k + \rho\,(D^u\,u_{k+1} - z_{k+1})  
			\end{array}
			\]
			where the update for $z_{k+1}$ can be done using the algorithm from \cite{duchi2008efficient}.
			
			As explained in the main manuscript, in practice, using update as part of an ADMM algorithm leads to difficulty enforcing the $\ell_1$ constraint in reasonable runtimes, and results in larger reconstruction error than the technique we have recommended.
			
			\section{Proof of technical results}

			\subsection{Proof of  Theorem \ref{main_theorem} }
			
			Note that the Lagrange dual function of the original problem is given
			by
			
			\[
			\begin{split}L\left(\lambda,\mu\right) & =\underset{u}{\text{minimize}}\left[-x^{T}u+\lambda\left(\| Du\|_{1}-c_{S}\right)+\mu\left(\| u\|_{2}^{2}\;-\,1\right)\right]\\
			& =\underset{u}{\text{minimize}}\left[-x^{T}u+\lambda\| Du\|_{1}+\mu\| u\|_{2}^{2}\right]-\mu-\lambda c_{S}\\
			& \lambda,\mu\geq0.
			\end{split}
			\]

			Next, define for fixed $\lambda,$ $\mu$ $\geq$ $0,$ the function $g_{\lambda,\mu}$$:$
			$\mathbb{R}^{S}$ $\longrightarrow$ $\mathbb{R}$ given by
			
			\begin{equation}
			g_{\lambda,\mu}\left(u\right)=-x^{T}u+\lambda\| Du\|_{1}+\mu\| u\|_{2}^{2}\label{e2}.
			\end{equation}

			From (\ref{e2}) we need to solve the following problem:
			
			\begin{equation}
			\underset{u}{\text{minimize}}\; g_{\lambda,\mu}\left(u\right)\label{e3} \, ,
			\end{equation}
			which can be rewriten as
			
			\[
			\begin{split}\, & \,\underset{u,z}{\text{minimize}}\left[-x^{T}u+\lambda\| z\|_{1}+\mu\| u\|_{2}^{2}\right]\\
			& \, s.t\qquad z=Du \, . 
			\\
			\end{split}
			\]
			This problem has the following Lagrangian:
			
			\[
			L_{\lambda,\mu}\left(z,u,\gamma\right)=-x^{T}u+\lambda\| z\|_{1}+\mu\| u\|_{2}^{2}+\gamma^{T}\left(Du-z\right) \, ,
			\]
			which is nicely separable in $u$ and $z$.
			%

			Let us now consider some special cases of $\mu$ and $\lambda$.
			First, if $\lambda = 0$ and $\mu = 0$, then clearly,
			
			\[
			\underset{z,u}{\text{min}}\; L_{\lambda,\mu}\left(z,u,\gamma\right)=-\infty\qquad\forall\gamma,
			\]
			%
			%
			%
			Second, if $\lambda = 0$ and $\mu > 0$, then
			
			\[
			\underset{u}{\text{min}}\left[-x^{T}u+\lambda\| Du\|_{1}+\mu\| u\|_{2}^{2}\right]=-\frac{1}{4\mu}x^{T}x.
			\]
			Next, if $\lambda$ $>$ $0$ and $\mu$ $=$ $0$, then
			
			\[
			\underset{z,u}{\text{min}}\; L_{\lambda,\mu}\left(z,u,\gamma\right)=-\infty\qquad\forall \gamma \quad \text{with}\quad D^{T}\gamma\neq x,
			\]
			and
			
			\[
			\underset{z,u}{\text{min}}\; L_{\lambda,\mu}\left(z,u,\gamma\right)=0\qquad\forall\gamma\; \text{with}\quad D^{T}\gamma=x\quad \text{and}\quad\|\gamma\|_{\infty}\leq\lambda.
			\]
			Thus
			\[
			\underset{u}{\text{min}}\left[-x^{T}u+\lambda\| Du\|_{1}\right]=\begin{cases}
			\begin{array}{ccc}
			-\infty & \text{if } x\notin Range\left(D^{T}\right)\\
			0 & \text{if  there exist }\gamma\; \text{with}\quad D^{T}\gamma=x\quad \text{and}\quad\|\gamma\|_{\infty}\leq\lambda.
			\end{array}\end{cases}
			\]
			
			Finally, let us now focus on $\mu$ $>$ $0$ or $\lambda$ $>$ $0.$ Then
			
			\[
			\underset{u}{\text{min}}\left[-x^{T}u+\mu\|u\|_{2}^{2}+\gamma^{T}Du\right]=-\frac{1}{4\mu}\| x-D^{T}\gamma\|_{2}^{2},
			\]
			while (see \cite{tibshirani2011solution})
			
			\[
			\underset{z}{\text{min}}\left[\lambda\|z\|_{1}-\gamma^{T}z\right]=\begin{cases}
			\begin{array}{ccc}
			0 & \text{if} & \|\gamma\|_{\infty}\leq\lambda \, ,\\
			-\infty &  & \text{otherwise}.
			\end{array}\end{cases}
			\]
			Hence, the dual problem to (\ref{e3}) is equivalent to
			
			\[
			\begin{split}\, & \,\underset{\gamma}{\text{minimize}}\quad\:\frac{1}{4\mu}\| x-D^{T}\gamma\|_{2}^{2}\\
			& \, \text{subject to} \qquad\|\gamma\|_{\infty}\leq\lambda \, . \\
			\\
			\end{split}
			\]
			But for $\mu$ $>$ $0$ fixed, this is equivalent to solving the problem
			
			\begin{equation}
			\begin{split}\, & \,\underset{\gamma}{\text{minimize}}\:\frac{1}{2}\| x-D^{T}\gamma\|_{2}^{2}\\
			& \, s.t\qquad\|\gamma\|_{\infty}\leq\lambda \, ,\\
			\\
			\end{split}
			\label{e4}
			\end{equation}
			which can be solved for every $\lambda$ $\geq$ $0$ using the solution
			path algorithm from \cite{tibshirani2011solution}.
			
			Let us denote by $\hat{\gamma}_{\lambda}$
			the solution to (\ref{e4}) for a fixed $\lambda.$ Therefore,
			
			\[
			L\left(\lambda,\mu\right)=-\frac{1}{4\mu}\| x-D^{T}\hat{\gamma}_{\lambda}\|_{2}^{2}-\mu-\lambda c_{S},
			\]
			which implies that the dual to the original problem becomes
			
			\begin{equation}
			\underset{\lambda,\mu\geq0}{\text{maximize}}\left[-\frac{1}{4\mu}\| x-D^{T}\hat{\gamma}_{\lambda}\|_{2}^{2}-\mu-\lambda c_{S}\right].
			\label{e5}
			\end{equation}
			Finally, recall from \cite{boyd2004convex} that any  $u^{*}$ solution to the original problem  must also solve
			
			\[
			u^{*}=\underset{u}{\text{arg min}}\left[-x^{T}u+\lambda^{*}\| Du\|_{1}+\mu^{*}\| u\|_{2}^{2}\right],
			\]
			for $\lambda^{*}$  and $\mu^{*}$ that are optimal for (\ref{e5}). However, the objective function in (\ref{e3}) is strictly convex since $\mu^{*}$
			$>$ $0,$ and so its solution $u^{*}$ is unique and also solves
			
			\[
			\begin{split}\, & \,\underset{u,z}{\text{minimize}}\left[-x^{T}u+\lambda^{*}\| z\|_{1}+\mu^{*}\|u\|_{2}^{2}\right]\\
			& \, \text{subejct to}\qquad z=Du.\\
			\\
			\end{split}
			\]
			The  KKT optimality conditions for this problem  imply that
			
			\[
			0=\left(\begin{array}{c}
			-x+2\mu^{*}u^{*}\\
			\lambda^{*}\alpha
			\end{array}\right)+\left(\begin{array}{c}
			D^{T}\gamma_{\lambda^{*}}\\
			-\gamma_{\lambda^{*}}
			\end{array}\right)
			\]
			for some $\alpha$ subgradient of the function $z$ $\rightarrow$
			$\| z\|_{1}$ at $z^{*}$ $=$ $D\, u^{*}.$ Therefore
			
			\[
			u^{*}=\frac{\left(x-D^{T}\hat{\gamma}_{\lambda^{*}}\right)}{\| x-D^{T}\hat{\gamma}_{\lambda^{*}}\|_{2}} \, ,
			\]
			and the result follows.
			
			$\qquad$$\qquad$$\qquad$$\qquad$$\qquad$$\qquad$$\qquad$$\qquad$$\qquad$$\qquad$$\qquad$$\qquad$$\qquad$$\qquad$$\qquad$$\qquad$$\qquad$$\qquad$$\:\:\:\,$$\square$

			\subsection{Proof of Theorem 2}

			Here we assume that data is generated as
			
			\[
			\underline{Y} =  d^{*}\,u^{*} \circ v^{*} \circ^{*} w  + \epsilon 
			\]
			and
			\[
			\| \hat{v} - v^{*} \|_2 < \frac{1}{\sqrt{2}},\,\,  \| \hat{w}- w^{*}\|_2 < \frac{1}{ \sqrt{2} }.
			\]
			Under these conditions we show that $\hat{u}$ defined as
			
			\[
			\begin{aligned}
			\hat{u} &=&& \underset{u\in\mathbb{R}^{S}}{\text{arg min}}
			& & - \underline{Y} \times_{2} \hat{v} \times_{3} \hat{w}  \\
			&&& \text{subject to} &&  \|u\|_{2}^{2} \leq 1,\,\,\,\,\|D^{(k_u+1)}u\|_1 \leq c_u
			\end{aligned}
			\]
			satisfies
			
			\[
			\text{P}\left( \frac{1}{2}\|u^{*} -\hat{u} \|_2^2   \leq \frac{1}{2} \frac{    \frac{c\,t}{d^{*}\sqrt{L}}  +  \frac{2\,c_u\,L^{k_u+1/2}}{d^{*}}     }{ \langle v^{*}, \hat{v}\rangle \,\langle w^{*}, \hat{w}\rangle - 2^{-1}  }  \right)  \geq  1- \sqrt{ \frac{2}{\pi} } \frac{\exp\left(-t^2/2\right)  }{t}
			- \frac{1}{L^{3/2}  \sqrt{\log L } }\sqrt{  \frac{2}{5\pi} }.
			\]
			The proof will then follow by an application of this claim after each block update, and applying the identity for the intersection of such events (see the main paper).

			To prove the claim above, we start by noticing that
			
			\[
			\begin{aligned}
			\hat{u} &=&& \underset{u\in\mathbb{R}^{S}}{\text{arg min}}
			& & -   (d^{*})^{-1}  \underline{Y} \times_{2} \hat{v} \times_{3} \hat{w}  \\
			&&& \text{subject to} &&  \|u\|_{2}^{2} \leq 1,\,\,\,\,\|D^{(k_u+1)}u\|_1 \leq c_u.
			\end{aligned}
			\]
			Next we use the notation $R$  for the row space of $D$: $R = row(D)$ and $R^{\perp} = null(D)$. Moreover, $P_V$ denotes the perpendicular projection onto the space $V$. Hence, by suboptimality,

			\begin{equation}
			\label{t2_1}
			\begin{array}{lll}
			\frac{1}{2} \|\hat{u} -u^{*} \|_2^2  &  \leq &   1 -  \hat{u}^{T} u^{*}+  \frac{ 1  }{d^{*}} \left(\underline{Y} \times_2 \hat{v} \times_3 \hat{w}\right)^{T}\left(  \hat{u} - u^{*} \right)  \\
			&  =  &   1 -  \hat{u}^{T} u^{*} +  \frac{ 1  }{d^{*}} \left( \left(d^{*}\,u^{*} \circ v^{*} \circ^{*} w  + \epsilon \right)  \times_2 \hat{v} \times_3 \hat{w}\right)^{T}\left( P_{R}  + P_{R^{\perp}}  \right)\left(  \hat{u} - u^{*} \right)\\
			&  =  & 1 -  \hat{u}^{T} u^{*} + \langle \hat{v}, v^{*} \rangle\,\langle \hat{w},w^{*} \rangle\,(u^{*})^T\left( \hat{u} - u^{*} \right) +  \frac{1}{d^{*}}\epsilon \times_2 \hat{v} \times_3 \hat{w}\left( P_R + P_{R^{\perp}}  \right)\left( \hat{u} - u^{*} \right).
			\end{array}
			\end{equation}
			
			Let us now bound  the terms in the expression above. First, let $a_1,\ldots,a_k$ be an orthonormal basis of $R^{\perp}$.  Then

			\begin{equation}
			\label{t2_2}
			\begin{array}{lll}
			\frac{1}{d^{*}} \sum_{j=1}^{k_u+1} \left(\epsilon \times_2 \hat{v} \times_3 \hat{w}\right)^{T}\,P_{R^{\perp}}\left( \hat{u} - u^{*} \right) & \leq & \frac{1}{d^{*}} \sum_{j=1}^{k_u + 1} \vert \left(\epsilon \times_2 \hat{v} \times_3 \hat{w}\right)^{T} a_j  \vert \,\|a_j\|_{\infty}\,2\,\sqrt{2}\  \\
			& \leq &   \frac{1}{d^{*}}\frac{c}{\sqrt{L}} \sum_{j= 1}^{k_u + 1}   \vert \left(\epsilon \times_2 \hat{v} \times_3 \hat{w}\right)^{T} a_j  \vert \\
			&\leq & \frac{1}{d^{*}}\frac{c\,(k_u + 1)\,t}{\sqrt{L}}
			\end{array}
			\end{equation}
			for some constant $c$ with probability at least
			
			\[
			1 -   (k_u+1) \sqrt{ \frac{2}{\pi}} \frac{  \exp\left(-t^2/2\right)}{t}.
			\]
			Here we have used Mill's inequality and the fact that we can take $\|a_j\|_{\infty} = O(L^{-1/2})$. The latter claim is immediate for $k_u = 0$. If $k_u = 1$ it can be proven as follows. First, we set $a_0 =  (1/\sqrt{L},\ldots,1/\sqrt{L})^T \in \mathbb{R}^L$. Then by the definition of $R^{\perp}$, an induction argument shows that
			
			\begin{equation}
			\label{big_o1}
			a_{1,k+1} = k\,a_{1,2}   - (k-1)a_{1,1}
			\end{equation}
			for  $k \in \{2,...,L-1\}$, where $a_{1,j}$  is the $j-$th coordinate of $a_{1}$.  But since $a_{1}$ is a unit vector, simple algebra yields
			
			\begin{equation}
			\label{big_o2}
			\begin{array}{lll}
			1  & = &    \left[ (L-1)^2 +   \frac{(L-2)(L-1)( 2(L-2)+1 )}{6} \right]a_{1,2}^2  +  \left[  \frac{(L-2)(L-1)( 2(L-2)+1 )}{6} \right]a_{1,1}^2   \\
			& & - 2\left[       \frac{(L-2)(L-1)( 2(L-2)+1 )}{6} + \frac{(L-2)(L-1)}{2} \right]a_{1,1}\,a_{1,2}.
			\end{array}
			\end{equation}
			Now, since $a_1$  and $a_0$  are orthogonal, we must have

			\begin{equation}
			\label{big_o3}
			0  = \frac{L(L-1)}{2}a_{1,2}   -   \left( \frac{(L-1)(L-2)}{2} -1 \right)a_{1,1}.
			\end{equation}
			Then the fact that $\|a_1\|_{\infty} = O(L^{-1/2})$ follows from (\ref{big_o1}), (\ref{big_o2}) and (\ref{big_o3}).
			
			Next we bound the term involving the projection operator onto the space $R$ in (\ref{t2_1}). By Holder's inequality,
			
			\[
			\begin{array}{lll}
			\frac{1}{d^{*}}\left(\epsilon \times_2 \hat{v} \times_3 \hat{w}\right)^T\, P_R \,\left( \hat{u} - u^{*} \right)  & \leq & \frac{1}{d^{*}}\left\|\left( \epsilon \times_2 \hat{v} \times_3 \hat{w}\right)^T (D^{(k_u+1)})^{-} \right\|_{\infty}\,\left( \|D^{(k_u+1)} \hat{u} \|_1  + \|D^{(k_u+1)} u^{*} \|_1  \right)
			\end{array}
			\]
			and hence, as in Corollary 4 from \cite{wang2014trend},  we find that
			
			\begin{equation}
			\label{t2_3}
			\text{P}\left(  \frac{1}{d^{*}}\epsilon \times_2 \hat{v} \times_3 \hat{w}\, P_R \,\left( \hat{u} - u^{*} \right) \leq  \frac{L^{k_u +1/2} \sqrt{\log(L)}  c_u}{d^{*}}  \right)  \geq  1  -  \frac{1}{L^{3/2}\sqrt{\log L }} \sqrt{\frac{2}{5\,\pi}}.
			\end{equation}
			
			On the other hand, by the Cauchy--Schwarz inequality, we have
			
			\begin{equation}
			\label{t2_4}
			\begin{array}{lll}
			1 -  \hat{u}^{T} u^{*} + \langle \hat{v}, v^{*} \rangle\,\langle \hat{w},w^{*} \rangle\,(u^{*})^T\left( \hat{u} - u^{*} \right) & =  &    \left(1 -  \langle \hat{v}, v^{*} \rangle\,\langle \hat{w},w^{*} \rangle \right)\left( \|u^*\|_2^2  - \langle \hat{u},u^{*} \rangle \right)\\
			& \leq &   \left(1 -  \langle \hat{v}, v^{*} \rangle\,\langle \hat{w},w^{*} \rangle \right) \|  \hat{u} - u^{*} \|_2^2.
			\end{array}
			\end{equation}
			Combining (\ref{t2_1}), (\ref{t2_2}), (\ref{t2_3}), (\ref{t2_4}), and proceeding in similar fashion for the other updates, the identity
			\[
			\text{P}(A \cap B \cap C) =  \text{P}(A)  \text{P}(B \mid A)  \text{P}(C \mid A \cap B)
			\]
			for any events $A$,$B$  and $C$ implies the result.
			
			For the case of multiple factors, we have the following result.  Suppose that the data is generated as
			
			\[
			\underbar{Y} \,\,=\,\, \sum_{j=1}^J d_j^*\,u^*_j\,\circ\,v^*_j\\,\circ,w^*_j \,\,+\,\, \,\bar{E}
			\]
			where $E$ is tensor of white noise. Suppose that we have current parameters estimates   of  $\{u^*_{j}\}_{j \neq j_0}$,  $\{v^*_{j}\}_{j }$,  $\{w^*_{j}\}_{j }$, $\{d^*_{j}\}_{j }$ which we denote by  $\{\hat{u}_{j}\}_{j \neq j_0}$,  $\{\hat{v}_{j}\}_{j }$,  $\{\hat{w}_{j}\}_{j }$, $\{\hat{d}_{j}\}_{j }$. 
			
			Let us now provide an error bound for the estimate of $u^*_{j_0}$ given all the other estimates. To that end, define 
			
			\[
			\hat{u}_{j_0} = \begin{aligned}
			& \underset{u\in\mathbb{R}^{L} }{\text{arg min}}
			& & \frac{1}{2}\left\| u-  \left(\,\underline{Y}\times_2 \hat{v}_{j_0}\times_3 \hat{w}_{j_0} - \sum_{j \neq j_0} \hat{d}_j\,(\hat{v}_{j_0})^T\,\hat{v}_{j}\,(\hat{w}_{j_0})^T\,\hat{w}_j\,\hat{u}_j\right)
			\right \|_{F}^{2} \\
			& \text{subject to}
			& &  \| D^{u}u\|_1 \,\leq\,c_u\\
			& & &  u^{T}u \,=\,1, 
			\end{aligned}
			\]
			and assume that $\| D^{u} u^*_{j_0}\|_1 \leq c_u$ and
			
			\[
			\| \hat{v}_{j_0} - v^*_{j_0}\|_2 < \frac{1}{\sqrt{2}},\,\,\,\,  \| \hat{w}_{j_0} - w^*_{j_0}\|_2 < \frac{1}{\sqrt{2}}
			\]
			
			This leads to the following lemma.
			\begin{lemma}
				Under the definitions just given,
				\[
				\text{P}\left( \frac{1}{2}\|u^{*}_{j_0} -\hat{u}_{j_0} \|_2^2   \leq \frac{1}{32}\left(  \frac{c\,t}{d_{j_0}^{*}\sqrt{L}}  +  \frac{2\,c_u\,L^{k_u+1/2}}{d_{j_0}^{*}}\right) +  U  \right)  \geq  1- \sqrt{ \frac{2}{\pi} } \frac{\exp\left(-t^2/2\right)  }{t}
				- \frac{1}{L^{3/2}  \sqrt{\log L } }\sqrt{  \frac{2}{5\pi} },
				\]
				where
				\[
				U  = \left\|\frac{1}{d_{j_0}^*} \sum_{j \neq j_0} \left(- \hat{d}_j\,\left(\hat{v}_j \cdot  \hat{v}_{j_0}\right)\left(\hat{w}_j \cdot  \hat{w}_{j_0} \right)\hat{u}_j  + d^*_j\,\left(v^*_j \cdot  \hat{v}_{j_0}\right)\left(w^*_j \cdot  \hat{w}_{j_0} \right)\,u^*_j   \right) \right\|_2.
				\]
			\end{lemma}
			\begin{proof}
				
				To show this we proceed as follows. We start noticing that, by sub-optimality,
				
				\[
				\begin{array}{lll}
				\frac{1}{2} \left\| \hat{u}_{j_0} - u^*_{j_0} \right\|_2^2  & \leq & \frac{1}{d_{j_0}^*}
				\left[ \sum_{j \neq j_0} \left(- \hat{d}_j\,\left(\hat{v}_j \cdot  \hat{v}_{j_0}\right)\left(\hat{w}_j \cdot  \hat{w}_{j_0} \right)\hat{u}_j  + d^*_j\,\left(v^*_j \cdot  \hat{v}_{j_0}\right)\left(w^*_j \cdot  \hat{w}_{j_0} \right)\,u^*_j   \right) \right] (-u^*_{j_0} + \hat{u}_{j_0}) \\
				& & +  1\,-\,   \hat{u}_{j_0}\cdot u^*_{j_0}  +     \frac{1}{d_{j_0}^*}
				\left[ \bar{E}\times_2\hat{v}_{j_0}\times_3 \hat{w}_{j_0}  +   d^*_{j_0}\left( v^*_{j_0}\cdot \hat{v}_{j_0}\right) \left( w^*_{j_0}\cdot \hat{w}_{j_0}\right)u^*_{j_0}\right] (-u^*_{j_0} + \hat{u}_{j_0}) \\
				& \leq &   \left\|\frac{1}{d_{j_0}^*} \sum_{j \neq j_0} \left(- \hat{d}_j\,\left(\hat{v}_j \cdot  \hat{v}_{j_0}\right)\left(\hat{w}_j \cdot  \hat{w}_{j_0} \right)\hat{u}_j  + d^*_j\,\left(v^*_j \cdot  \hat{v}_{j_0}\right)\left(w^*_j \cdot  \hat{w}_{j_0} \right)\,u^*_j   \right) \right\|_2 \\
				& & +  1\,-\,   \hat{u}_{j_0}\cdot u^*_{j_0}  +     \frac{1}{d_{j_0}^*}
				\left[ \bar{E}\times_2\hat{v}_{j_0}\times_3 \hat{w}_{j_0}  +   d^*_{j_0}\left( v^*_{j_0}\cdot \hat{v}_{j_0}\right) \left( w^*_{j_0}\cdot \hat{w}_{j_0}\right)u^*_{j_0}\right] (-u^*_{j_0} + \hat{u}_{j_0}) \\
				\end{array}
				\]
				and hence the claim follows for the case $J = 1$.
				
			\end{proof}

			\section{Discussion and extensions}
			\label{connections}
			
			\subsection{Further connections with existing work}
			\label{tpca}

			In recent years, many different efforts have been made to apply the ideas of sparse regression and sparse matrix decomposition to the context of higher-order tensors.  Our paper has shown that structured penalties from the generalized-lasso class can offer significant modeling benefits when the underlying factors are piecewise constant or smooth.  Moreover, our main result shows that the factors can be efficiently computed by a coordinate-wise optimization routine, exploiting results on the solution path of the dual problem for the generalized lasso.  Both the simulated and real examples have shown the power of the approach.
			
			Our general framework has applications across a wide class of problem formulations for analyzing multi-way data.  In this section, we describe some connections with other existing methods.  We also describe how orthogonality constraints can be imposed in our approach.
			
			Recall that in the usual PCA framework we are
			given samples $x_{1}, \ldots, x_{m} \in \mathbb{R}^{J}$, and the task is to find a unit vector $a \in \mathbb{R}^{J}$
			such that the points $x_{1}^{T}a, \ldots, x_{m}^{T}a,$ on the real
			line have the largest possible variance. The problem can be stated in matrix notion as
			\[
			\underset{\| a\|_{2}=1}{\text{maximize}}\quad a^{T}X^{T}Xa \, .
			\]
			
			By imposing L1 constraints, the authors of \cite{jolliffe2003modified} propose to sacrifice the variance explained in order to gain interpretability. The resulting problem, called SCoTLASS, is
			\[
			\text{maximize}\quad a^{T}X^{T}Xa\qquad \text{subejct to}\quad\|a\|_{2}^{2}\leq1,\quad\| a\|_{1}\leq c \, .
			\]
			
			More generally, the authors of \cite{lu2008mpca} consider Multilinear Principal Component
			Analysis of Tensor Objects (MPCA).  This is defined for a set of tensors
			$A_{1}, A_{2}, \ldots, A_{M} \in \mathbb{R}^{I_{1}\times I_{2}\times\cdots I_{N}}$
			by the solution to the following problem:
			\[
			\left\{ U_{m}^{*T},m=1,..M\right\} =\underset{\left\{ U_{m}^{T},m=1,..M\right\} }{\text{arg max}}\Psi_{b},
			\]
			where $B_{m} = A_{m} \times_{1} U_{1}^{T} \times_{2} \cdots \times_{M} U_{M}^{T}$,
			$m = 1, \ldots, M$ and $\Psi_{b}$ is their total scatter. The motivation is to perform feature extraction by determining a multilinear projection that captures most of the original tensorial input variation. The fitting algorithm proceeds by iteratively decomposing the original problem to a series of multiple projection sub-problems.
			
			Combining the regularization idea of SCoTLASS with MPCA, we can formulate the penalized MPCA problem as
			\begin{equation}
			\label{penalized_mpca}
			\begin{aligned}
			& \underset{u,v}{\text{maximize}} && \underset{k=1}{\overset{m}{\sum}}\mid\left(\bar{X}-X_{k}\right)\bar{\times}_{1}u\bar{\times}_{2} v \mid^{2}\qquad\qquad\qquad\qquad\;\\
			& \text{subject to}  && P_{S}\left(u\right)\,\leq\, c_{S},\quad P_{T}\left(v\right)\,\leq\, c_{T},\quad\,\| u\|_{2}^{2},\,\| v\|_{2}^{2}\,\leq\,1,
			\end{aligned}
			\end{equation}
			where $\bar{X}$ is the sample mean of the training data $X_1, X_2, \ldots, X_k$. The solution to this problem allows us to project the training data into a lower-dimensional space in a way that maximizes the variance explained while retaining structural constraints in the projection space. The key point is that we can use the rank-1 PTD algorithm to solve (\ref{penalized_mpca}), since it can be  verified  that (\ref{penalized_mpca})  is equivalent to
			\[
			\begin{aligned}
			& \underset{u,v,w}{\text{maximize}} && Y\bar{\times}_{1}u\bar{\times}_{2}v\bar{\times}_{3}w\\
			& \text{subject to} && \, P_{S}\left(u\right)\,\leq\, c_{S}\quad P_{T}\left(v\right)\,\leq\, c_{T}\\
			& && \| u\|_{2}^{2}\,\leq\,1,\quad\| v\|_{2}^{2}\,\leq\,1,\quad\| w\|_{2}^{2}\,\leq\mathbf{\,}1,
			\end{aligned}
			\]
			where $\underline{Y}$ $\in$ $\mathbb{R}^{S\times T\times m}$ is
			a tensor satisfying that $Y_{s,t,k}$ $=$ $\left(X_{k}\right)_{s,t}$.  This connection is, in fact, analogous to the connection established in \cite{witten2009penalized} between SCoTLASS and the PMD algorithm.

			\subsection{Orthogonal factors}
			\label{orthogonal_factors}
			
			We now return to the multiple-factor decomposition proposed in the main paper. Given and input data tensor $Y$, we seek to find a decomposition as the sum of $k$ rank-1 tensors, as in the Parafac model. We proposed an algorithm to find such a representation based on our algorithm for rank-1 PTD, but there were no constraints regarding the orthogonality of the vectors involved in the representation.   Orthogonality is a natural constraint in factor-type models, and it is often imposed in tensor decompositions; see \cite{cichocki2013tensor,kolda2009tensor,de2000multilinear}. In the framework of matrix decomposition, the authors of \cite{witten2009penalized} explored an approach to obtain multiple rank-1 factors that were sparse and whose vectors were unlikely to be correlated. However, no formal guarantee was provided that the output vectors would be orthogonal. Here we fill that gap and provide a simple method for finding factors whose vectors are orthogonal and satisfy structural constraints, including sparsity.
			
			Suppose that we are given $k$ rank-1 tensors that approximate $\underline{Y}$. At the $k+1$ step, we  try to find a rank-1 tensor that best approximates the current tensor of residuals. This is done by solving an optimization problem whose objective function is the Frobenius norm of the residual, with structural constraints specified by the chosen penalties. If we also impose the additional constraint of orthogonality, then the update for $u_{k+1}$  can be written as the solution of a problem of the form
			\[
			\begin{aligned}
			& \underset{u}{\text{minimize}}
			& &
			u^{T}x \\
			& \text{s.t}
			& &  \| u\|_{2}^{2}\,\leq\,1,\quad\|Du\|_{1}\leq c,\quad u^{T}u_{j}\,=\,0\quad\forall j=1,...,k-1 \, .
			\end{aligned}
			\]
			We can further rewrite this as
			\begin{equation}
			\begin{aligned}
			\label{sparse_problem}
			& \underset{u}{\text{minimize}}
			& &
			\theta^{T}\tilde{x} \\
			& \text{s.t}
			& &  \| \theta\|_{2}^{2}\,\leq\,1,\quad\|\tilde{D}\theta\|_{1}\leq c,
			\end{aligned}
			\end{equation}
			where the matrix $\tilde{D}$  equals the product of $D$ and a  matrix whose columns form a basis of the orthogonal complement of the space spanned by $u_1$,....,$u_{k-1}$;  see \cite{witten2009penalized}. Hence, we can use our rank-1 PTD algorithm to find sparse orthogonal Parafac decompositions.
			
			The orthogonality constraint imposes additional computational burdens.  As the authors of \cite{arnold2014efficient} point out, problems  of the form (\ref{sparse_problem}) can be solved efficiently if the matrix $\tilde{D}$ is sparse.  This can happen if the vectors are $u_1$,  $u_2$,...,$u_{k-1}$ are highly sparse.  If, on the other hand, $\tilde{D}$  is not sparse, then (\ref{sparse_problem}) can be solve via its dual, using a projected-Newton method similar to the recent algorithm in \cite{wang2014trend}.
			
			\subsection{Multilinear regression}
			\label{regression}
			
			Here we show how some of the basic ideas in multilinear
			regression are related to our methodology. See \cite{zhao2013higher} for a discussion of multilinear regression.  A more general
			approach for tensor regression is discussed in \cite{cichocki2013tensor}.
			
			Motivated by the statistical setting in \cite{banerjee2004hierarchical}, and by the discussion of tensor regression given in \cite{cichocki2013tensor}, we consider the problem of finding single-factor representations of $\underline{X} \in \mathbb{R}^{S\times T\times J}$ and $\underline{Y} \in \mathbb{R}^{S\times T}$ such that
			\begin{equation}
			\label{Regression_problem}
			\underline{X}\approx g\, p\circ q\circ a,\quad\underline{Y}\approx d\, p\circ q,\qquad g,d\in\mathbb{R},\: p\in\mathbb{R}^{S},\: q\in\mathbb{R}^{T},\: a\in\mathbb{R}^{J}.
			\end{equation}
			
			The intuition behind (\ref{Regression_problem}) corresponds to a problem in which, for every time point $t$ and location
			$s$, there exists an observation $y_{s,t}$ and a vector of covariates
			$x_{s,t,:}$.  Hence it is natural to impose the constraint that
			the one-factor representations of $X$ and $Y$ have common vectors
			associated with time and location. The difficulty of this problem
			lies in the fact that we need to simultaneously approximate $X$ and $Y$ by
			the representations in (\ref{Regression_problem}). Below we formally state a version of this problem, incorporating some additional constraints that are merely for identifiability purposes.
			
			\[
			\begin{aligned}
			& \underset{p\in\mathbb{R}^{S},q\in\mathbb{R}^{T},a\in\mathbb{R}^{J},g,d\in\mathbb{R}}{\text{minimize}} && \|\underline{X}-g\, p\circ q\circ a\|_{F}^{2}+\|\underline{Y}-d\, p\circ q\|_{F}^{2}\\
			& \text{subject to }&& P_{S}\left(p\right)\,\leq\, c_{S}\quad P_{T}\left(q\right)\,\leq\, c_{T}\quad P_{J}\left(a\right)\,\leq\, c_{J}\quad\\
			&  && \| p\|_{2}^{2} \, \leq \, 1\quad\| q\|_{2}^{2} \, \leq \,1\quad\| a\|_{2}^{2} \, \leq \, 1. \quad
			\end{aligned}
			\]

			Clearly the objective function in (\ref{Regression_problem}) is a quadratic form for each
			of $p$, $q$, and $a$ individually, while holding the other terms fixed. This can make the solving the problem complicated. Alternatively, we can  try to maximize the product of the terms $\prec\underline{X},p\circ q\circ a \succ$ and $\prec\underline{Y},p\circ q\succ$, as observed by \cite{zhao2013higher}. But we notice the following elementary inequality:
			\[
			\begin{array}{ccc}
			2\langle\underline{Y},p\circ q\rangle\langle\underline{X},p\circ q\circ a\rangle & \leq & \frac{\left(\langle\underline{X},p\circ q\circ a\rangle+\langle\underline{Y},p\circ q\rangle\right)^{2}}{2}\qquad\qquad\quad\\
			& \leq & \langle\underline{X},p\circ q\circ a\rangle{}^{2}+\langle\underline{Y},p\circ q\rangle{}^{2}.
			\end{array}
			\]
			Hence, it makes sense to solve the problem
			\begin{equation}
			\label{final_regression_problem}
			\begin{aligned}
			& \underset{p\in\mathbb{R}^{S},q\in\mathbb{R}^{T},a\in\mathbb{R}^{J},g,d\in\mathbb{R}}{\text{minimize}} && \langle\underline{X},p\circ q\circ a\rangle+\prec\underline{Y},p\circ q\rangle\\
			& \text{subject to }&& P_{S}\left(p\right)\,\leq\, c_{S}\quad P_{T}\left(q\right)\,\leq\, c_{T}\quad P_{J}\left(a\right)\,\leq\, c_{J}\quad\\
			&  && \| p\|_{2}^{2} \, \leq \, 1\quad\| q\|_{2}^{2} \, \leq \,1\quad\| a\|_{2}^{2} \, \leq \, 1 \quad
			\end{aligned}
			\end{equation}
			which has an  trilinear obective function in $(p,q,a)$ Thus, we can
			try to solve (\ref{final_regression_problem}) by using coordinate
			wise optimization, taking advantage of our previous developments.
			
			Although we do not include simulations for problem  (\ref{final_regression_problem})  in our experiments section, our investigations suggest that combining the information of both the predictors $\underline{X}$ and the response $\underline{Y}$  can provide better results than just fitting a PTD on $\underline{Y}$ and a PMD on  $\underline{X}$ separately.

			\subsection{Extensions to Tucker models}
			\label{extensions}
			
			Up until now we have being interested in Parafac models, which are special cases of general Tucker model.  A penalized Tucker model was proposed in \cite{cichocki2013tensor} in which the goal is to maximize with respect to $U^{\left(n\right)}$ $\in$
			$\mathbb{R}^{I_{n}\times J_{n}},$ $n = 1, \ldots, N$ the cost function
			\[
			D_{F}\left(Y\| G,\left\{ U\right\} \right)=\|\underline{Y}-G\times\left\{ U\right\} \|_{F}^{2}+\underset{n}{\sum}\alpha_{n}C_{n}\left(U^{\left(n\right)}\right),
			\]
			with penalties $C_{1}, \ldots, C_{n}$ on $U^{\left(1\right)}, \ldots, U^{\left(n\right)}$
			respectively and positive parameters $\alpha_1, \ldots, \alpha_n$.
			
			We provide some insight on a penalized Tucker problem with generalized-lasso
			penalties on the columns of each $U^{\left(n\right)}$. For simplicity
			of notation, we assume $N=3$, $J_{n}=2$, and $n = \{1, 2, 3\}$.  Our formulation of the problem becomes
			\begin{equation}
			\label{penalized_tucker_model}
			\begin{aligned}
			& \underset{u_{:1}^{\left(1\right)},...,u_{:1}^{\left(N\right)}}{\text{minimize}} && \|\underline{Y}-\underset{j_{1},j_{2},j_{3}}{\sum}d_{j_{1}j_{2}j_{3}}\, u_{:j_{1}}^{\left(1\right)}\circ u_{:j_{2}}^{\left(2\right)}\circ u_{:j_{3}}^{\left(3\right)}\|_{F}^{2}\\
			& \text{subject to }&& P_{n}\left(u_{:j}^{\left(n\right)}\right)\,\leq\, c_{n}\quad\forall n\in\left\{ 1,2,3\right\} ,\quad j\in\left\{ 1,2\right\} \\
			&&& \| u_{:j}^{\left(n\right)}\|_{2}^{2}=1\quad\forall n\in\left\{ 1,...,N\right\} \quad j\in\left\{ 1,2\right\} \, .
			\end{aligned}
			\end{equation}
			This can be rewritten as an optimization problem whose objective function is  is linear on each $u_{:j_{i}}^{\left(i\right)}$ when the other variables are fixed, and convex on each $d_{j_{1}j_{2}j_{3}}$ when every other variable is fixed. Hence, we can use an algorithm similar to our rank-1 PTD procedure based on coordinate wise optimization.
			
			There is yet a different way to think about Tucker models. In this
			class of problems the core tensor is considered random, and the interest
			lies in reconstructing the matrices $U^{\left(n\right)},$ $n= 1,\ldots,N$, which
			are assumed to be invertible.  The model is written as $\underline{Y} = Z \times \left\{ U\right\} $ where $Z$ is an array of independent
			standard normal entries; see \cite{hoff2011separable}. There, the authors proved that
			$cov\left(\underline{Y}\right) = \Sigma_{1} \circ \Sigma_{2} \circ \cdots \circ \Sigma_{N}$, with $\Sigma_{n}= U^{\left(n\right)} \left(U^{\left(n\right)}\right)^{T}$.  The matrices $U^{\left(n\right)}$ introduce covariance
			structure to the model.
			
			Given samples $\underline{Y}_{1}, \ldots, \underline{Y_{n}}$ we would like
			to estimate $\Sigma_{1},$..., $\Sigma_{n}$.  Hence we form the following
			problem:
			\[
			\underset{\Sigma_{n} \in S^+}{\text{maximize}}\quad \log\, P\left(\underline{Y}_{1},...,\underline{Y_{n}}\mid\Sigma_{1},...,\Sigma_{n}\right)-\underset{n}{\overset{}{\sum}}\lambda_{n}P\left(\Sigma_{n}\right) \, ,
			\]
			where the constraint is the set of non-negative definite matrices.  This formulation appeared in \cite{hoff2011separable}, but without the penalties. Similar formulations
			including penalties can be found in \cite{leng2012sparse} and \cite{friedman2010applications}. In fact, a
			coordinate descent type of algorithm can be used that is similar
			to the one proposed in \cite{hoff2011separable}, but that solves every subproblem with
			methods described in \cite{leng2012sparse} and \cite{friedman2010applications}.

			\section{Simulation details}
			
			In our set of experiments we considered   5  different hidden rank-1 tensors  constructed  as $u \circ v \circ w$  where the vectors $u$, $v$  and $w$  are described below.  The notation $\{x\}_i^j$ indicates that components $i$ through $j$ of the vector are all equal to the value $x$.
			
			\paragraph{Structure 1}
			\begin{itemize}
				\item $u = \{1,1,1,-1,-1,-1,0,0,0,0\}$.
				\item $v = \{0\}_1^{100},\{1\}_{101}^{500},\{0\}_{501}^{1000}$.
				\item $w$ $=$   $\{-1\}_1^{100},\{0\}_{101}^{200},\{1\}_{201}^{400}$.
			\end{itemize}

			\paragraph{Structure 2}
			\begin{itemize}
				\item $u$ $=$  $\{0,0,0,-1,-1,-1,0,0,0,0\}$.
				\item  $v = \{v_i\}_{i=1}^{1000}$ with  $v_i = \text{cos}\left(12\,\pi\,\frac{(i-1)}{999}\right)$ for $i =  1,2, \ldots, 1000$.
				\item $w = \{w_i\}_{i=1}^{400}$ with  $w_i = \text{cos}\left(9\,\pi\,\frac{(i-1)}{399}\right)$ for $i =   1,2, \ldots, 400$.
			\end{itemize}

			\paragraph{Structure 3} 
			\begin{itemize}
				\item $u$ $=$  $\{0,0,0,0,-1,-1,1,1,1,1\}$.
				\item $v = \{v_i\}_{i=1}^{1000}$ with $v_i = \left(\frac{(i-1)}{999}-0.7\right)^2 + \left(\frac{(i-1)}{999}\right)^2$  for $i =   1,2, \ldots, 1000$.
				\item  Define $w_{i}^{\prime}$  $=$ $\frac{i-1}{399}$ for $i = 1,\dots, 400$. Then, set  $w_{i}$ $=$   $ w_{i}^{\prime}\,\left(0.05-w_{i}^{\prime}\right)$  for  $i = 1, \ldots, 200$  and  $w_{i}$ $=$   $\left(w_{i}^{\prime}\right)^2 $  for  $i = 201,\dots, 400$.
			\end{itemize}

			\paragraph{Structure 4}
			\begin{itemize}
				\item $u$ $=$   $\{0,0,0,0,0,1,1,1,1,1\}$  
				\item   Define   $v_{i}^{\prime}$  $=$  $\frac{i-1}{999}$ for $i = 1,\dots, 1000$. Then, \\
				$v_{i}$  $=$    $cos(\pi\,v_i^{\prime}) + .65$.   
				\item  $w$  $=$  $ \{0\}^{100},\{1\}_{101}^{150},\{0\}_{151}^{300},\{1\}_{301}^{350},\{0\}_{351}^{400}$.
			\end{itemize}

			\paragraph{Structure 5}
			\begin{itemize}
				\item $u$ $=$  $ \{-1,-1,0,0,1,1,1,-1,-1,-1\}$.
				\item $v$  has $80\%$ of its entries equal to zero and the remanining $20\%$  are random numbers drawn from a standar normal distribution.
				\item $w$  has $92.5\%$ of its entries equal to zero and the remanining $7.5\%$  are random numbers drawn from a standar normal distribution.
			\end{itemize}

			\section{Real data examples additional details}

			\subsection{Flu hospitalizations}
			
			Our flu example uses aggregate, non-identifiable hospitalization records from each of the eight largest counties in Texas from January 1, 2003 to December 30, 2009.  Our data-use agreement does not permit dissemination of these hospital records.  We also use data on temperature and air quality (particulate matter) in these counties, which can be obtained directly from CDC Wonder (http://wonder.cdc.gov/).
			
			\subsection{Motion capture }
			To construct the tensors involved in the five task considered, we use the variables: the second coordinate for root (variable 2), the first coordinate for upperback (variable 10),  the first coordinate for upperneck (variable 19), the first coordinate for head (variable 22), the second coordinate for rhumerus (variable 28), rradius (variable 30), the second coordinate for lhumerus (variable 40), lradius (variable 42), the second coordinate for lhand  (variable 44), lfingers (variable 45), rtibia  (variable 52), ltibia (variable 59).
			
			For task 138--story we use videos corresponding to subject 138 in the moCap repository. Videos 11-14 are used to construct the training tensor while 15-18 are used to build the test tensor.      
			
			To build task 107 walking we use videos from subject 107. For training we use videos 1-4 for training while videos 5-8 are used for testing.
			
			For task 09-run we use videos corresponding to subject 9. Videos 1-4 are used for training, and videos 5-8 are used for testing.
			
			To construct task 138 marching we take videos from subject 138. For training we use videos 1-4 for training while videos 5-8 are used for testing.
			
			Finally, for task 126, the training set is built using videos 1,3,6,8 while the test set uses videos 2,4,7,9.

\end{document}